\documentclass[runningheads, envcountsame]{llncs}
\usepackage[T1]{fontenc}
\usepackage{graphicx}
\usepackage[utf8]{inputenc}
\usepackage{amsmath}
\usepackage{amssymb}
\usepackage{amsfonts}
\usepackage[dvipsnames]{xcolor}
\usepackage[title]{appendix}
\usepackage{comment}
\usepackage{hyperref}
\usepackage{mathtools}
\usepackage{tikz}
\usepackage{stmaryrd}
\usepackage[capitalize]{cleveref}
\usepackage{eqparbox}
\usepackage{caption}
\usepackage{subcaption}

\usetikzlibrary{arrows,automata,shapes,positioning}

\newcommand{\true}{\mathit{true}}
\newcommand{\false}{\mathit{false}}

\newcommand{\AP}{\mathit{AP}}

\newcommand{\intern}{\mathit{int}}
\newcommand{\spawn}{\mathit{spawn}}
\newcommand{\call}{\mathit{call}}
\newcommand{\ret}{\mathit{ret}}
\newcommand{\callret}{\mathit{callRet}}
\newcommand{\labelend}{\mathit{end}}
\newcommand{\leaf}{\mathit{leaf}}
\newcommand{\moves}{\mathsf{Moves}}
\newcommand{\dir}{\mathsf{Dir}}
\newcommand{\ExGraphs}{\mathsf{ExGraphs}}
\newcommand{\ExTrees}{\mathsf{ET}}
\newcommand{\treelabels}{\mathit{TL}} 

\newcommand{\Succ}{\mathit{succ}}
\newcommand{\Sub}{\mathit{Sub}}

\newcommand{\fp}{\mathit{fp}}
\newcommand{\ar}{\mathit{ar}}
\newcommand{\assignment}{\mathit{as}}
\newcommand{\ap}{\mathit{ap}}
\newcommand{\globalsucc}{g}
\newcommand{\globalpred}{\uparrow}
\newcommand{\abstractsucc}{a}
\newcommand{\caller}{-}
\newcommand{\parent}{p}
\newcommand{\child}{c}
\newcommand{\reachable}{r}

\newcommand{\NN}{\mathbb{N}}

\newcommand{\PTIME}{\mathsf{PTIME}}
\newcommand{\EXPTIME}{\mathsf{EXPTIME}}

\newcommand{\logicname}{NTL}

\newcommand\mybox[2][]{\mathrel{\overset{\scriptsize{\eqmakebox[#1]{}}}{#2}}}
\newcommand\oversetbox[3][]{\mathrel{\overset{\scriptsize{\eqmakebox[#1]{#2}}}{#3}}}

\let\llncsproof\proof
\renewcommand{\proof}[1][]{%
	\ifx!#1!\else\renewcommand{\proofname}{#1}\fi
	\llncsproof
}

\allowdisplaybreaks

\begin{document}
	\bibliographystyle{plainurl}

	\title{A Navigation Logic for Recursive Programs with Dynamic Thread Creation}
	%
	%
	\author{Roman Lakenbrink \and
		Markus Müller-Olm \and
		Christoph Ohrem \and
		Jens Gutsfeld}
	\authorrunning{R. Lakenbrink et al.}
	%
	\institute{Department of Computer Science, University of Münster, Germany}
	\maketitle              
	\begin{abstract}
Dynamic Pushdown Networks (DPNs) are a model for multithreaded programs with recursion and dynamic creation of threads.
In this paper, we propose a temporal logic called \logicname\ for reasoning about the call- and return- as well as thread creation behaviour of  DPNs.
Using tree automata techniques, we investigate the model checking problem for the novel logic and show that its complexity is not higher than that of LTL model checking against pushdown systems despite a more expressive logic and a more powerful system model.
The same holds true for the satisfiability problem when compared to the satisfiability problem for a related logic for reasoning about the call- and return-behaviour of pushdown systems.
Overall, this novel logic offers a promising approach for the verification of recursive programs with dynamic thread creation.
\end{abstract}


	\section{Introduction}

Model Checking is an established technique for the verification of hardware and software systems.
Conceptually, it consists of checking whether a property given in a specification logic holds for a model of a system.
While logics such as LTL or CTL and finite Kripke models were considered early on \cite{Clarke1981,Lichtenstein1985}, later more expressive logics as well as infinite state systems have been studied.
The use of pushdown systems, for instance, allows for a more precise analysis of recursive software systems due to the presence of a call stack of a program while still retaining a decidable model checking problem against LTL specifications \cite{Bouajjani1997}.
In the context of pushdown systems, an example of a logic more expressive than LTL is the logic CaRet \cite{Alur2004} which extends LTL by operators for non-regular properties of the call and return behaviour of pushdown systems.
This extension does not lead to increased complexity for the model checking problem against pushdown systems compared to LTL.

However, there are even more powerful system models than pushdown systems for which model checking of variants of LTL is decidable.
In this paper, we consider Dynamic Pushdown Networks (DPNs) \cite{Bouajjani2005}, a model for software systems that cannot model only recursion, but also multithreading with dynamic thread creation.
So far, the model checking problem for DPNs has only been considered for single indexed LTL, a variant of LTL for multithreaded systems in which an LTL formula is assigned to each thread \cite{Song2015}.
Here, we consider the model checking problem for a more expressive logic.
More specifically, we propose a fixpoint calculus with CaRet-like operators for the verification of DPNs via model checking. 
The logic allows to specify non-regular properties concerning the call and return behaviour of the different execution threads of a DPN.
Unlike CaRet, it can additionally specify properties concerning the thread-spawn behaviour of programs.
For example, consider a scenario where a program has a method for bookkeeping information about spawned threads and it is required that new threads be only spawned from this method in order to keep the bookkeeping consistent.
The property $\mathcal{G}^\reachable (\bigcirc^\child \psi \rightarrow \mathcal{F}^\caller \textit{pr})$
specified in our new logic expresses that in all positions of all threads (expressed through the modality $\mathcal{G}^\reachable$), new threads fulfilling the property $\psi$ are only spawned (expressed through $\bigcirc^\child \psi$) when the procedure $\textit{pr}$ is in the call stack (expressed through $\mathcal{F}^\caller \textit{pr}$).
This formalises the requirement.
Properties regarding such relationships between parent and child threads cannot be expressed in the variant of LTL from \cite{Song2015} or other specification logics for DPNs we are aware of.
Our logic thus constitutes the first specification logic able to reason about the thread spawning behaviour of DPNs.

\textbf{Contributions and structure of the paper.}
After introducing some notation and results (\cref{sec:preliminaries}), we present a semantics for DPNs based on graphs (\cref{sec:graphsemantics}).
As our first main contribution, we then introduce a novel specification logic called Navigation Temporal Logic (\logicname) with the ability to reason about the call/return and thread creation behaviour of DPNs (\cref{sec:logic}).
We discuss some example properties and applications of our logic in \cref{sec:examples}.
Towards algorithmic verification, we then switch from a semantics based on graphs to a semantics based on trees  (\cref{sec:treesemantics}).
As our second main contribution, we then investigate the model checking and satisfiability problems for the new logic (\cref{sec:decisionproblems}).
In particular, we show that the model checking problem is decidable in time exponential in the size of the specification and polynomial in the size of the system model, i.e.\ the same as for LTL model checking against pushdown systems, and that the satisfiability problem is solvable in time exponential in the size of the specification, i.e.\ the same as for the satisfiability problem for $\mathit{VP}$-$\mu$-$\mathit{TL}$ \cite{Bozzelli2007}, a temporal logic subsuming CaRet and subsumed by our logic.
For both problems, we establish matching lower bounds.
\cref{sec:conclusion} concludes the paper.
Due to lack of space, some technical proofs can be found in an appendix.

\textbf{Related work.}
There are several specification logics related to the logic we present in this paper.
The temporal logic LTL was considered for model checking finite state systems \cite{Lichtenstein1985} as well as pushdown systems \cite{Bouajjani1997}.
For pushdown models, CaRet was developed with different successor types that allow the inspection of the call and return behaviour of the system \cite{Alur2004}.
Also, variants of CaRet have been studied in the literature \cite{Alur2006a,Alur2008,Bozzelli2011,Gutsfeld2018}.
As mentioned, CaRet is one inspiration for the logic presented in this paper and we adopt and complement its successor types in our logic.
Other inspirations are the linear time $\mu$-calculus from \cite{Vardi1988} and the logic $\mathit{VP}$-$\mu$-$\mathit{TL}$ from \cite{Bozzelli2007}.
In these logics, fixpoint operators can be used to express arbitrary $\omega$-regular (resp. $\omega$-visibly pushdown) properties on paths, which makes them more expressive than LTL and CaRet, respectively.
From these logics, we take fixpoint operators.
There is also a plethora of work on dynamic pushdown networks.
The model was first introduced in \cite{Bouajjani2005}.
Different methods for reachability analysis of DPNs have been proposed \cite{Bouajjani2005,Lammich2009}.
Additionally, different variants of the model were investigated.
\cite{Nordhoff2013} and \cite{Lammich2013} consider variants of DPNs that communicate via locks.
Another variant is the model of Dynamic Networks of concurrent pushdown systems from \cite{Atig2011} in which threads can communicate via global variables.
However, none of the above works on DPNs is concerned with model checking.
The only approach to model checking DPNs we are aware of consists of checking different variants of DPNs against a variant of LTL called single indexed LTL \cite{Song2015,Diaz2018}.
Compared to NTL, this variant cannot specify properties concerning the call and return behaviour of a thread or the relationship between different threads.
In \cref{sec:examples}, we show that single indexed LTL can be embedded into NTL.
	\section{Preliminaries}\label{sec:preliminaries}

Without further ado, we introduce tools and notation used throughout the paper.
This section can be skipped on first reading and be consulted for reference later.

\textbf{Trees.}
An $\NN_0$-\textit{tree} $T$ is a prefix-closed subset of $\NN_0^*$, i.e.\ for all nodes $t \in \NN_0^*$ and directions $d \in \NN_0$, $t \cdot d \in T$ implies $t \in T$.
Moreover, we require that $t\cdot d\in T$ for some $d\in\NN_0$ implies $t\cdot d'\in T$ for all $d'\leq d$.
We call an element $t \in T$ a \textit{node} of $T$ with special node $\varepsilon$, which we call the \textit{root}.
A node of the form $t \cdot d$ is called a \textit{child} of $t$ and $t$ is called the \textit{parent} of $t\cdot d$.
Additionally, for sequences $w \in \NN_0^*$, we call $t \cdot w$ a \textit{descendant} of $t$ and $t$ an \textit{ancestor} of $t \cdot w$.
Nodes $t \in T$ that have no children are called \textit{leaves}.
A \textit{path} in a tree $T$ is a finite or infinite sequence $t_0 t_1 \dots$ of nodes such that $t_0 = \varepsilon$ and for all $i \in \mathbb{N}_0$, $t_{i+1}$ is a child of $t_i$.
An $\NN_0$-tree that is a subset of $\{0,1\}^*$ is also called a \textit{binary tree}.
In this case, we call a node of the form $t \cdot 0$ the \textit{left child} of $t$ and a node of the form $t\cdot 1$ the \textit{right child} of $t$.
Let $\Sigma$ be a finite set of labels and $\ar \colon \Sigma \to \{0,1,2\}$ be a function assigning an arity to each of these labels. 
A $(\Sigma, \ar)$-\textit{labelled binary tree} is a pair $(T,l)$ such that $T$ is a binary tree and $l \colon T \to \Sigma$ is a labelling function such that each node $t\in T$ has exactly $\ar(l(t))$ children. 

\textbf{2-way alternating tree automata.}
For a finite set $X$, let $\mathcal{B}^+(X)$ be the set of positive boolean combinations over $X$, i.e.\ boolean formulae built with elements of $X$, conjunction and disjunction.
For $Y \subseteq X$ and $\vartheta \in \mathcal{B}^+(X)$, we say that $Y$ \textit{satisfies} $\vartheta$ iff assigning the value $\true$ to the elements of $Y$ and $\false$ to the elements of $X \setminus Y$ makes the formula $\vartheta$ true.
Let $\dir = \{0,1,\varepsilon,\uparrow\}$ be the set of moves in the tree with directions $0$ for the left child, $1$ for the right child, $\varepsilon$ for standing still and $\uparrow$ for moving upwards.
We define $u \cdot \varepsilon = u$ and $u\cdot d\ \cdot \uparrow\ = u$ for all $u \in \{0,1\}^*$ 
and $d \in \{0,1\}$.
A \textit{2-way alternating tree automaton} (2ATA) \cite{Vardi1998} over $(\Sigma,\ar)$-labelled binary trees is a tuple $\mathcal{A} = (Q,q_0,\rho,\Omega)$ where $Q$ is a finite set of states, $q_0 \in Q$ is an initial state, $\rho \colon Q \times \Sigma \to \mathcal{B}^+(\dir \times Q)$ is a transition function and $\Omega \colon Q \to \{0,\dots,k\}$ is a priority mapping.
The size $|\mathcal{A}|$ of a 2-way alternating tree automaton is defined as the sum of the sizes of its constituents.
We sometimes also refer to the size of individual constituents of an automaton.
In particular, we refer to the number of states, i.e.\ $|Q|$, and the size of the acceptance condition, i.e.\ $k$.
If the transition function of a 2-way alternating tree automaton uses only symbols from $\{0,1\}$ instead of $\dir$ and additionally maps all nodes $t$ either to $\true$ or to disjunctions over conjunctions that consist of exactly one pair $(d,q)$ for each $d<\ar(l(t))$, it is called a \textit{nondeterministic parity tree automaton} (NPTA).

For a $(\Sigma,\ar)$-labelled binary tree $\mathcal{T}=(T,l)$, a node $t\in T$ and a state $q\in Q$, a \textit{$(t,q)$-run} of $\mathcal{A}$ over $\mathcal{T}$ is a pair $(T_r,r)$ such that $T_r$ is an $\mathbb{N}_0$-tree and $r \colon T_r \to T \times Q$ assigns a pair of a node of $T$ and a state of $\mathcal{A}$ to all nodes in $T_r$.
Additionally, $(T_r,r)$ has to satisfy the following conditions: 
(i) $r(\varepsilon) = (t,q)$ and 
(ii) for all nodes $y \in T_r$ with $r(y) = (x,s)$ and $\rho(s,l(x)) = \vartheta$, there is a set $Y \subseteq \dir \times Q$ satisfying $\vartheta$ and for all $(d,s') \in Y$, there is $n \in \mathbb{N}_0$ such that $y \cdot n \in T_r$ and $r(y \cdot n) = (x \cdot d,s')$. 
In particular, for all leaves $y\in T_r$ with $r(y)=(x,s)$, we thus require $\rho(s,l(x))=\true$.
A $(t,q)$-run $(T_r,r)$ is \textit{accepting} iff on each infinite path in $T_r$ the lowest priority occurring infinitely often is even.
If $\mathcal{A}$ is a nondeterministic parity tree automaton, a minimal set $Y$ satisfying $\rho(s,l(x))$ in the above definition moves to each child of the current node $x$ in the tree $T$. 
For an $(\varepsilon,q)$-run, we can thus simply identify $T$ with $T_r$ and consider a map $r_A\colon T\to Q$ as an $(\varepsilon,q)$-run over $\mathcal{A}$.
The set of nodes $t\in T$ such that there is an accepting $(t,q)$-run of $\mathcal{A}$ over $\mathcal{T}$ is denoted by $\mathcal{L}_{q}^\mathcal{T}(\mathcal{A})$.
We say that $\mathcal{A}$ \textit{accepts} a tree $\mathcal{T}$ iff there is an accepting $(\varepsilon,q_0)$-run of $\mathcal{A}$ over $\mathcal{T}$.
The set of trees accepted by $\mathcal{A}$ is denoted by $\mathcal{L}(\mathcal{A})$.
We use the following theorems:
\begin{proposition}[\cite{Vardi1998}]\label{prop:treedealternation}
	For every 2ATA $\mathcal{A}$, there is an equivalent NPTA $\mathcal{A}'$.
	The number of states in $\mathcal{A}'$ is at most exponential in the number of states of $\mathcal{A}$ and the size of the acceptance condition of $\mathcal{A}'$ is linear in the size of the acceptance condition of $\mathcal{A}$.
\end{proposition}
\begin{proposition}[\cite{Emerson1988,Kupferman1998,Pnueli1989}]\label{prop:treeemptiness}
	The emptiness problem for NPTA can be solved in time polynomial in the number of states and exponential in the size of the acceptance condition.
\end{proposition}
\begin{proposition}\label{prop:treeintersection}
	\begin{enumerate}
		\item[(i)] For any two NPTA $\mathcal{A}_1$ and $\mathcal{A}_2$, there is a NPTA $\mathcal{A}$ with $\mathcal{L}(\mathcal{A}) = \mathcal{L}(\mathcal{A}_1) \cap \mathcal{L}(\mathcal{A}_2)$.
		\item[(ii)] If either acceptance condition is trivial, the size of $\mathcal{A}$ is in $\mathcal{O}(|\mathcal{A}_1| \cdot |\mathcal{A}_2|)$.
	\end{enumerate}
	
\end{proposition}
\cref{prop:treeintersection} (i) can be found e.g. in \cite{Loding2021}.
For (ii), a straightforward product construction can be used and yields an automaton of the size claimed.

\textbf{Dynamic Pushdown Networks.}
Let $\AP$ be a set of atomic propositions, $\Gamma$ be a finite set of stack symbols and $\bot \notin \Gamma$ be a special bottom of stack symbol.
A \textit{Dynamic Pushdown Network} (DPN) \cite{Bouajjani2005} is a tuple $\mathcal{M} = (S,s_0,\gamma_0,\Delta,L)$ where $S$ is a finite set of control locations, $s_0 \in S$ is an initial control location, $\gamma_0\in\Gamma$ is an initial stack symbol and $L \colon S \times \Gamma \to 2^\AP$ is a labelling function.
The transition relation $\Delta = \Delta_I \overset{.}{\cup} \Delta_C \overset{.}{\cup} \Delta_R \overset{.}{\cup} \Delta_S$ is a finite set of \textit{internal} rules ($\Delta_I$), \textit{calling} rules ($\Delta_C$), \textit{returning} rules ($\Delta_R$) and \textit{spawning rules} ($\Delta_S$).
Internal rules $s\gamma \rightarrow s'\gamma' \in \Delta_I \subseteq S\Gamma \times S\Gamma$, call rules $s\gamma \rightarrow s'\gamma'\gamma'' \in \Delta_C \subseteq S\Gamma \times S\Gamma^2$ or returning rules $s\gamma \rightarrow s' \in \Delta_R \subseteq S\Gamma \times S$ enable transitions of a single pushdown process in control location $s$ with top of stack $\gamma$ to the new control location $s'$ with new top of stack $\gamma'$, $\gamma'\gamma''$ and $\varepsilon$, respectively.
A spawning rule $s\gamma \rightarrow s'\gamma' \triangleright s_n \gamma_n \in \Delta_S \subseteq S\Gamma \times S\Gamma \times S\Gamma$ is an internal rule with the additional side effect of spawning a new process in control location $s_n$ and stack content $\gamma_n$.
We formally develop a semantics for DPNs in \cref{sec:graphsemantics}.
	
	\section{Graph Semantics of Dynamic Pushdown Networks}\label{sec:graphsemantics}

The semantics of DPNs is often defined as an \textit{interleaving semantics}.
In such semantics, a configuration of a DPN is a collection of local configurations of the underlying pushdown systems representing the currently active threads.
A step in this semantics consists of a step of one of the active threads, possibly adding a configuration of a new thread to the collection.
This way, the semantics accurately reflects different interleavings of the steps of the threads issued by arbitrary schedulers, hence the name.
For our intents, interleaving semantics has some drawbacks, however.
First, an encoding of the intersection problem for contextfree languages is often possible in interleaving semantics, which leads to undecidability of the investigated verification problem.
Second, we are mostly interested in temporal properties of individual threads, not necessarily temporal properties of interleavings.
This is because the behaviour of a thread is in most cases independent of what types of steps other threads currently make in a specific interleaving.
Third, we want to reason about the parent-child relationship of threads which is lost in most formalisations of interleaving semantics.

We thus instead adopt a semantics based on graphs.
Intuitively, in an \textit{execution graph}, each thread is modelled by a linear sequence of positions connected by $\intern$-, $\call$- and $\ret$-edges based on the types of transitions taken in the thread.
In order to model the parent-child relationship between threads, a position where a spawn-transition is taken is connected to the first position of the spawned thread via a $\spawn$-edge.
This is analogous to the notion of \textit{action trees} from \cite{Lammich2009,Gawlitza2011}.
Additionally, similar to \textit{nested words} \cite{Alur2006}, calls and their matching returns are connected via \textit{nesting edges}.
We formalise these graphs in the next paragraph.

\textbf{Execution graphs.}
Let $\moves = \{\intern,\call,\ret,\spawn\}$ be the set of moves for dynamic pushdown networks, $V$ be a set of nodes, $l \colon V\to 2^\AP$ be a labelling function, $\rightarrow^d\ \subseteq V^2$ be a transition relation for all $d\in\moves$ and $\curvearrowright\ \subseteq V^2$ be a nesting relation.
For $x,y \in V$, we call $y$ a \textit{($d$)-successor} of $x$ and $x$ a \textit{($d$)-predecessor} of $y$ if $x\rightarrow^d y$ for some $d\in\moves$. 
A tuple $G = (V, l, (\rightarrow^d)_{d\in\moves}, \curvearrowright)$ is called an \emph{execution graph}, iff the following conditions hold:
\begin{enumerate}
	\item Every node has exactly one predecessor with respect to $\bigcup \{\rightarrow^d \mid d\in\moves \}$ except for a special node $v_0$ without predecessor.
	\item For all $x\in V$ we have $(v_0,x)\in (\bigcup \{\rightarrow^d \mid d\in\moves \})^*$.
	\item Every node either has 
	(a) exactly one $\intern$-successor and at most one $\spawn$-successor, 
	(b) exactly one $\call$-successor,
	(c) exactly one $\ret$-successor or
	(d) no successors.
	\item On every finite path starting in $v_0$ or a node with a $\spawn$-predecessor and following only $\moves \setminus \{\spawn\}$-successors, the number of $\call$-moves on that path is greater than or equal to the number of $\ret$-moves on that path.
	\item For all $x\in V$ having a $\call$-successor, let $A_x$ be the set of nodes $y\neq x$ such that there is a path $\pi$ from $x$ to $y$ following only $\moves \setminus \{\spawn\}$-successors where the number of $\call$-moves on $\pi$ is equal to the number of $\ret$-moves on $\pi$.
	Then we have $x\curvearrowright y$ for a node $y\in V$ iff $y$ is a node in $A_x$ such that the witnessing path has minimal length.
\end{enumerate}
The set of execution graphs is denoted by $\ExGraphs$.

An example of an execution graph can be found in \cref{fig:executiongraph}.
In this example, a main thread spawns two additional threads.
Additionally, there are two nested procedure calls in the main thread and one procedure call in a spawned thread.
\begin{figure}
	\centering
	\resizebox{.8\textwidth}{!}{
		\begin{tikzpicture}
			\node[draw,circle,line width=1pt,label=$v_0$] at (0,0) (n11) {};
			\node[draw,circle,line width=1pt] at (1,0) (n12) {};
			\node[draw,circle,line width=1pt] at (2,0) (n13) {};
			\node[draw,circle,line width=1pt] at (3,0) (n14) {};
			\node[draw,circle,line width=1pt] at (4,0) (n15) {};
			\node[draw,circle,line width=1pt] at (5,0) (n16) {};
			\node[draw,circle,line width=1pt] at (6,0) (n17) {};
			\node[draw,circle,line width=1pt] at (7,0) (n18) {};
			\node[draw,circle,line width=1pt] at (8,0) (n19) {};
			\node[] at (9,0) (dots1) {\dots};
			
			\node[draw,circle,line width=1pt] at (1,-1.5) (n21) {};
			\node[draw,circle,line width=1pt] at (2,-1.5) (n22) {};
			\node[draw,circle,line width=1pt] at (3,-1.5) (n23) {};
			\node[draw,circle,line width=1pt] at (4,-1.5) (n24) {};
			\node[] at (5,-1.5) (dots2) {\dots};
			
			\node[draw,circle,line width=1pt] at (6,-1) (n31) {};
			\node[draw,circle,line width=1pt] at (7,-1) (n32) {};
			\node[draw,circle,line width=1pt] at (8,-1) (n33) {};
			\node[draw,circle,line width=1pt] at (9,-1) (n34) {};
			\node[] at (10,-1) (dots3) {\dots};
			
			\path[draw,->,line width=1pt] (n11) to node[midway,below]{\footnotesize$\mathit{int}$} (n12);
			\path[draw,->,line width=1pt] (n12) to node[midway,below]{\footnotesize$\mathit{int}$} (n13);
			\path[draw,->,line width=1pt] (n13) to node[midway,below]{\footnotesize$\mathit{call}$} (n14);
			\path[draw,->,line width=1pt] (n14) to node[midway,below]{\footnotesize$\mathit{call}$} (n15);
			\path[draw,->,line width=1pt] (n15) to node[midway,below]{\footnotesize$\mathit{ret}$} (n16);
			\path[draw,->,line width=1pt] (n16) to node[midway,below]{\footnotesize$\mathit{int}$} (n17);
			\path[draw,->,line width=1pt] (n17) to node[midway,below]{\footnotesize$\mathit{int}$} (n18);
			\path[draw,->,line width=1pt] (n18) to node[midway,below]{\footnotesize$\mathit{ret}$} (n19);
			\path[draw,->,line width=1pt] (n19) to node[midway,below]{\footnotesize$\mathit{int}$} (dots1);

			\path[draw,->,dashed,out=40,in=140,line width=1pt] (n14) to (n16);
			\path[draw,->,dashed,out=40,in=140,line width=1pt] (n13) to (n19);
			
			\path[draw,->,line width=1pt] (n12) to node[pos=.6,right]{\footnotesize$\mathit{spawn}$} (n21);
			
			\path[draw,->,line width=1pt] (n21) to node[midway,below]{\footnotesize$\mathit{int}$} (n22);
			\path[draw,->,line width=1pt] (n22) to node[midway,below]{\footnotesize$\mathit{call}$} (n23);
			\path[draw,->,line width=1pt] (n23) to node[midway,below]{\footnotesize$\mathit{ret}$} (n24);
			\path[draw,->,line width=1pt] (n24) to node[midway,below]{\footnotesize$\mathit{int}$} (dots2);
			
			\path[draw,->,dashed,out=40,in=140,line width=1pt] (n22) to (n24);
			
			\path[draw,->,line width=1pt] (n17) to node[pos=.7,right]{\footnotesize$\mathit{spawn}$} (n31);
			
			\path[draw,->,line width=1pt] (n31) to node[midway,below]{\footnotesize$\mathit{int}$} (n32);
			\path[draw,->,line width=1pt] (n32) to node[midway,below]{\footnotesize$\mathit{int}$} (n33);
			\path[draw,->,line width=1pt] (n33) to node[midway,below]{\footnotesize$\mathit{int}$} (n34);
			\path[draw,->,line width=1pt] (n34) to node[midway,below]{\footnotesize$\mathit{int}$} (dots3);
		\end{tikzpicture}
	}
	\caption{Example of an execution graph. Labelled edges represent edges $\rightarrow^d$ for $d \in \moves$ and dashed edges represent nesting edges $\curvearrowright$.}
	\label{fig:executiongraph}
\end{figure}
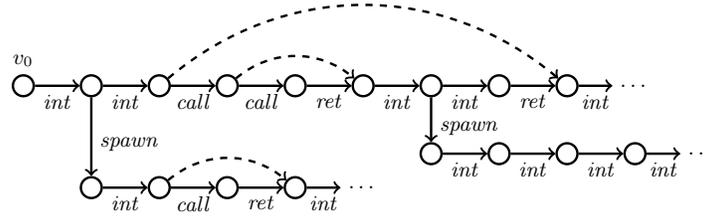

\textbf{Graph semantics.}
In most cases, we care only about graphs generated by a given DPN instead of arbitrary execution graphs.
This is formalised in the \textit{graph semantics of DPN}.
In the definition of this semantics, we make use of \textit{configurations} of the processes of a given DPN $\mathcal{M} = (S,s_0,\gamma_0,\Delta,L)$.
Formally, a configuration of a pushdown process of $\mathcal{M}$ is a pair $c = (s,u)$ where $s \in S$ is a control location and $u \in \Gamma^*\bot$ is a stack content ending in $\bot$.
We define successor relations on configurations corresponding to the different types of transition rules of DPNs.
For this purpose, let $c = (s,u)$, $c' = (s',u')$ and $c'' = (s'',u'')$ be configurations.
We call $c'$ an internal successor of $c$, denoted by $c \rightarrow_{\intern} c'$, if there is a transition $s\gamma \rightarrow s'\gamma' \in \Delta_I$ and $u = \gamma w$, $u' = \gamma'w$ for some stack content $w \in \Gamma^*\bot$.
We call $c'$ a call successor of $c$, denoted by $c \rightarrow_{\call} c'$, if there is a transition $s\gamma \rightarrow s'\gamma'\gamma'' \in \Delta_C$ and $u = \gamma w$, $u' = \gamma'\gamma''w$ for some stack content $w \in \Gamma^*\bot$.
We call $c'$ a return successor of $c$, denoted by $c \rightarrow_{\ret} c'$, if there is a transition $s\gamma \rightarrow s' \in \Delta_R$ and $u = \gamma w$, $u' = w$ for some stack content $w \in \Gamma^*\bot$.
Finally, we call $c'$ a successor of $c$ with spawned process $c''$, denoted $c \rightarrow c' \triangleright c''$, if there is a transition $s\gamma \rightarrow s'\gamma' \triangleright s'' \gamma'' \in \Delta_S$ and $u = \gamma w$, $u' = \gamma'w$ for some stack content $w \in \Gamma^*\bot$ as well as $u'' = \gamma'' \bot$.
Using the notion of configurations and the successor relations just introduced, we now define the graph semantics.
We say that an execution graph  $(V, l, (\rightarrow^d)_{d\in\moves}, \curvearrowright)$ \emph{is generated} by $\mathcal{M}$ if there is an assignment $\assignment \colon V \to S\times\Gamma^*\bot$ satisfying (i) $\assignment(v_0) = (s_0, \gamma_0\bot)$ and (ii) for all $x \in V$, we have $l(x) = L(s,\gamma)$ where $\assignment(x)=(s,\gamma w)$ for some control location $s\in S$, stack symbol $\gamma\in \Gamma$ and stack content $w\in \Gamma^*\bot$ and
	\begin{itemize}
		\item if $x$ has only one $d$-successor $y$ with $d\in\{\intern,\call,\ret\}$, then $\assignment(x) \rightarrow_d \assignment(y)$,
		\item if $x$ has an $\intern$-successor $y$ and a $\spawn$-successor $z$, then $\assignment(x) \rightarrow \assignment(y) \triangleright \assignment(z)$ and
		\item if $x$ has no successor, 
		then $\assignment(x)$ has no successor.
	\end{itemize}
The set of execution graphs generated by $\mathcal{M}$ is denoted by $\llbracket \mathcal{M} \rrbracket$.

\textbf{Successor functions.}
On execution graphs $G = (V, l, (\rightarrow^d)_{d\in\moves}, \curvearrowright)$, we define multiple successor functions $\Succ^G_\globalsucc$, $\Succ^G_{\globalpred}$, $\Succ^G_\abstractsucc$, $\Succ^G_\caller$, $\Succ^G_\parent$ and $\Succ^G_\child$ with signature $V \rightsquigarrow V$, i.e. partial functions from $V$ to $V$.
The first four of these successor functions come from logics like CaRet \cite{Alur2004} and allow us to progress single threads and their call-return behaviour in different ways.
The latter two functions are new and give means to reason about the thread spawning behaviour of DPNs.
For $x\in V$, the functions are defined as follows:
\begin{itemize}
	\item The \emph{global successor} $\Succ^G_\globalsucc(x)$ of $x$ is the $\intern$-, $\call$- or $\ret$-successor of $x$, if it exists, and undefined otherwise.
	
	\item The \emph{global predecessor} $\Succ^G_\globalpred(x)$ of $x$ is the $\intern$-, $\call$- or $\ret$-predecessor of $x$, if it exists, and undefined otherwise.
	
	\item The \emph{abstract successor} $\Succ^G_\abstractsucc(x)$ of $x$ is the node $y$ with $x\curvearrowright y$ or $x \rightarrow^\intern y$, if it exists, and undefined otherwise.
	
	\item The \emph{caller} $\Succ^G_\caller(x)$ of $x$ is the node $y$ with a $\call$-successor $y'$ such that there is a path from $y'$ to $x$ following abstract successors, if it exists, and undefined otherwise.
	
	\item The \emph{parent} $\Succ^G_\parent(x)$ of $x$ is the node $y$ with a $\spawn$-successor $z$ such that there is a path from $z$ to $x$ with only $\moves\setminus\{\spawn\}$-transitions, if it exists, and undefined otherwise.
	
	\item The \emph{child} $\Succ^G_\child(x)$ of $x$ is the $\spawn$-successor of $x$, if it exists, and undefined otherwise.
\end{itemize}
We illustrate these successor functions for parts of the execution graph from \cref{fig:executiongraph} in \cref{fig:stacksuccessors} and \cref{fig:spawnsuccessors}.
Abstract successors (seen in dashdotted red in \cref{fig:stacksuccessors}) follow the execution of a procedure on the same stack level and skip over executions of additional procedures via nesting edges.
If a procedure is left in the next step, i.e.\ if the next step is a return, the abstract successor is undefined.
Callers (seen in dotted blue in \cref{fig:stacksuccessors}) are defined if the stack level is at least one in a configuration and move to the latest previous call on a lower stack level.
Parents (seen in dotted green in \cref{fig:spawnsuccessors}) are defined in every branch of an execution graph representing a thread except for the thread starting in $v_0$ and move to the position in the graph where the current thread was spawned.
Children (seen in dashdotted yellow in \cref{fig:spawnsuccessors}) are defined only if the current thread currently executes a spawn transition and move to the initial position of the spawned thread.
\begin{figure}
	\begin{subfigure}{.6\textwidth}
		\centering
		\resizebox{\textwidth}{!}{
		\begin{tikzpicture}
			\node[draw,circle,line width=1pt,label=$v_0$] at (0,0) (n11) {};
			\node[draw,circle,line width=1pt] at (1,0) (n12) {};
			\node[draw,circle,line width=1pt] at (2,0) (n13) {};
			\node[draw,circle,line width=1pt] at (3,0) (n14) {};
			\node[draw,circle,line width=1pt] at (4,0) (n15) {};
			\node[draw,circle,line width=1pt] at (5,0) (n16) {};
			\node[draw,circle,line width=1pt] at (6,0) (n17) {};
			\node[draw,circle,line width=1pt] at (7,0) (n18) {};
			\node[draw,circle,line width=1pt] at (8,0) (n19) {};
			\node[] at (9,0) (dots1) {\dots};
			
			\node[gray!30] at (1,-1.5) (n21) {...};
			
			\node[gray!30] at (6,-1.5) (n31) {...};
			
			\path[draw,->,dashdotted,red,line width=1pt] (n11) to node[midway,above,black]{\footnotesize$\mathit{int}$} (n12);
			\path[draw,->,dashdotted,red,line width=1pt] (n12) to node[midway,above,black]{\footnotesize$\mathit{int}$} (n13);
			\path[draw,->,line width=1pt] (n13) to node[midway,above]{\footnotesize$\mathit{call}$} (n14);
			\path[draw,->,line width=1pt] (n14) to node[midway,above]{\footnotesize$\mathit{call}$} (n15);
			\path[draw,->,line width=1pt] (n15) to node[midway,above]{\footnotesize$\mathit{ret}$} (n16);
			\path[draw,->,dashdotted,red,line width=1pt] (n16) to node[midway,above,black]{\footnotesize$\mathit{int}$} (n17);
			\path[draw,->,dashdotted,red,line width=1pt] (n17) to node[midway,above,black]{\footnotesize$\mathit{int}$} (n18);
			\path[draw,->,line width=1pt] (n18) to node[midway,above]{\footnotesize$\mathit{ret}$} (n19);
			\path[draw,->,dashdotted,red,line width=1pt] (n19) to node[midway,above,black]{\footnotesize$\mathit{int}$} (dots1);
			
			\path[draw,->,gray!30,line width=1pt] (n12) to node[midway,right]{\footnotesize$\mathit{spawn}$} (n21);
			
			\path[draw,->,gray!30,line width=1pt] (n17) to node[midway,right]{\footnotesize$\mathit{spawn}$} (n31);
			
			\path[draw,->,out=230,in=310,densely dotted,NavyBlue,line width=1pt] (n15) to (n14);
			\path[draw,->,out=230,in=310,densely dotted,NavyBlue,line width=1pt] (n14) to (n13);
			\path[draw,->,out=230,in=310,densely dotted,NavyBlue,line width=1pt] (n16) to (n13);
			\path[draw,->,out=230,in=310,densely dotted,NavyBlue,line width=1pt] (n17) to (n13);
			\path[draw,->,out=230,in=310,densely dotted,NavyBlue,line width=1pt] (n18) to (n13);
			
			\path[draw,->,out=50,in=130,dashdotted,red,line width=1pt] (n13) to (n19);
			
			\path[draw,->,out=50,in=130,dashdotted,red,line width=1pt] (n14) to (n16);			
		\end{tikzpicture}
		}
		\caption{Abstract successors (\textcolor{red}{red, dashdotted}) and callers (\textcolor{NavyBlue}{blue, dotted}). Irrelevant edges are gray and some internal edges coinciding with abstract successors are omitted to improve readability.}
		\label{fig:stacksuccessors}
	\end{subfigure}
	\hfill
	\begin{subfigure}{.35\textwidth}
		\centering
		\resizebox{\textwidth}{!}{
		\begin{tikzpicture}
			\node[draw,circle,line width=1pt,label=$v_0$] at (0,0) (n11) {};
			\node[draw,circle,line width=1pt] at (1,0) (n12) {};
			\node[] at (2,0) (n13) {...};
			
			\node[draw,circle,line width=1pt] at (1,-2) (n21) {};
			\node[draw,circle,line width=1pt] at (2,-2) (n22) {};
			\node[draw,circle,line width=1pt] at (3,-2) (n23) {};
			\node[draw,circle,line width=1pt] at (4,-2) (n24) {};
			\node[] at (5,-2) (dots2) {\dots};

			\path[draw,->,line width=1pt] (n11) to node[midway,above]{\footnotesize$\mathit{int}$} (n12);
			\path[draw,->,line width=1pt] (n12) to node[midway,above]{\footnotesize$\mathit{int}$} (n13);
			
			\path[draw,->,dashdotted,YellowOrange,line width=1pt] (n12) to node[midway,left,black]{\footnotesize$\mathit{spawn}$} (n21);
			
			\path[draw,->,line width=1pt] (n21) to node[midway,below]{\footnotesize$\mathit{int}$} (n22);
			\path[draw,->,line width=1pt] (n22) to node[midway,below]{\footnotesize$\mathit{call}$} (n23);
			\path[draw,->,line width=1pt] (n23) to node[midway,below]{\footnotesize$\mathit{ret}$} (n24);
			\path[draw,->,line width=1pt] (n24) to node[midway,below]{\footnotesize$\mathit{int}$} (dots2);
			
			\path[draw,->,dashed,out=40,in=140,gray!30,line width=1pt] (n22) to (n24);
			
			\path[draw,->,out=45,in=315,densely dotted,ForestGreen,line width=1pt] (n21) to (n12);
			\path[draw,->,out=90,in=315,densely dotted,ForestGreen,line width=1pt] (n22) to (n12);
			\path[draw,->,out=90,in=315,densely dotted,ForestGreen,line width=1pt] (n23) to (n12);
			\path[draw,->,out=90,in=315,densely dotted,ForestGreen,line width=1pt] (n24) to (n12);
		\end{tikzpicture}
		}
		\caption{Parents (\textcolor{ForestGreen}{green, dotted}) and children (\textcolor{YellowOrange}{yellow, dashdotted}). Irrelevant edges are gray to improve readability.}
		\label{fig:spawnsuccessors}
	\end{subfigure}
	\caption{Successor types in parts of the execution graph from \cref{fig:executiongraph}.}
\end{figure}
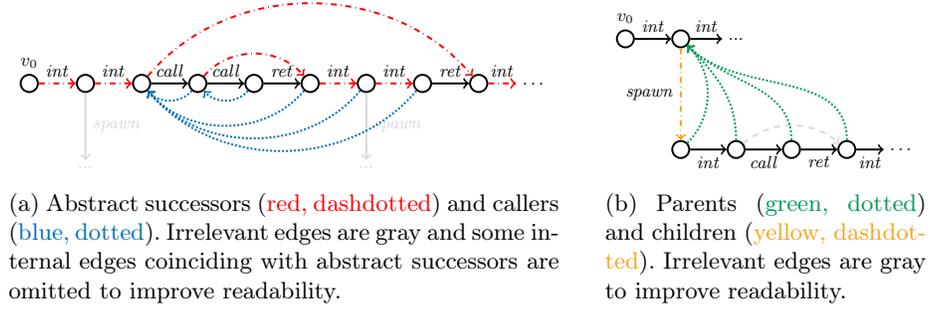
	\section{A Navigation Logic for Dynamic Pushdown Networks}\label{sec:logic}

\textbf{Syntax.}
We now define the new logic Navigation Temporal Logic (\logicname) for expressing properties of execution graphs.
As mentioned in the introduction, we have three main inspirations.
From the logics CaRet \cite{Alur2004} and $\mathit{VP}$-$\mu$-$\mathit{TL}$ \cite{Bozzelli2007}, we take different next operators inspecting the call and return behaviour of a thread.
We complement these by additional next operators expressing parent and child relationships between different processes.
From the linear time $\mu$-calculus \cite{Vardi1988} and logics like $\mathit{VP}$-$\mu$-$\mathit{TL}$ \cite{Bozzelli2007}, we take fixpoint operators for additional expressivity beyond LTL modalities.
First, we define the syntax of \logicname.
\begin{definition}[Syntax of \logicname]
	The syntax of \logicname\ formulae is defined by
	\begin{align*}
		\varphi ::=&\quad \ap \mid \lnot \varphi \mid \varphi_1 \lor \varphi_2 \mid X \mid \bigcirc^f \varphi \mid \mu X. \varphi
	\end{align*}
	where $\ap \in AP$ is an atomic proposition, $X$ is a fixpoint variable and $f\in \{\globalsucc,\globalpred,\abstractsucc,\caller,\parent,\child\}$ is a successor type.
\end{definition}

An \logicname\ formula $\varphi$ is called \emph{closed}, if every fixpoint variable $X$ is bound in $\varphi$, i.e. it only appears in a subformula of the form $\mu X.\psi$.
A formula $\varphi$ is called \emph{well-formed}, if every fixpoint variable $X$ occurring in $\varphi$ (i) is bound by only one fixpoint formula which we then denote by $\fp(X)$, (ii)  appears only in the scope of an even number of negations inside $\fp(X)$ and (iii) is in scope of at least one next operator inside $\fp(X)$.
We use $\Sub(\varphi)$ for the set of subformulae of a formula $\varphi$.
The size $|\varphi|$ of a formula $\varphi$ is defined as the number of its distinct subformulae.
We also need a notion of substitution: $\varphi[\varphi'/X]$ is the formula that is obtained from $\varphi$ by replacing every occurrence of the fixpoint variable $X$ with $\varphi'$.

Let us explain the intuition behind each construct.
Atomic formulae $\ap$ express that $\ap \in \AP$ holds in the current node of the graph.
Next operators $\bigcirc^f \varphi$ can be used to navigate and express that the corresponding successor exists in the current node and additionally satisfies $\varphi$.
Negation and disjunction are interpreted as usual.
Finally, we have fixpoint variables $X$ and least fixpoint operators $\mu X. \varphi$ for more involved properties.
Intuitively, $\mu X. \varphi$ is the least fixpoint of a function that \textit{unrolls} the formula by replacing $\mu X. \varphi$ with $\varphi[\mu X. \varphi/X]$.

We use some common syntactic sugar such as $\true \equiv \ap \lor \lnot \ap$, $\false \equiv \lnot \true$, $\varphi_1\land \varphi_2\equiv\lnot(\lnot \varphi_1\lor \lnot \varphi_2)$, $\varphi_1 \rightarrow \varphi_2 \equiv \lnot \varphi_1 \lor \varphi_2$, $\varphi_1 \leftrightarrow \varphi_2 \equiv (\varphi_1 \rightarrow \varphi_2) \land (\varphi_2 \rightarrow \varphi_1)$ and $\nu X.\varphi\equiv \lnot \mu X.\lnot \varphi[\lnot X/X]$. 
We also introduce a dual operator $\bigcirc^{\overline{f}}\varphi\equiv \lnot \bigcirc^f\lnot \varphi$ of $\bigcirc^f \varphi$ for each successor type $f\in \{\globalsucc,\globalpred,\abstractsucc,\caller,\parent,\child\}$ which is needed for a special form in the next paragraph. 
It is necessary to explicitly define these dual operators since $\lnot\bigcirc^f\varphi$ is not equivalent to $\bigcirc^f \lnot\varphi$ as the corresponding successors can be undefined for certain nodes in an execution graph.
Unlike $\bigcirc^f \varphi$, $\bigcirc^{\overline{f}} \varphi$ is equivalent to $\true$ for nodes that do not have an $f$-successor.
We also introduce some variants of LTL modalities as derived operators.
In particular we use $\varphi_1\ \mathcal{U}^f \varphi_2 \equiv \mu X. (\varphi_2 \lor (\varphi_1 \land \bigcirc^f X))$, $\mathcal{F}^f \varphi \equiv \true\ \mathcal{U}^f \varphi$ and $\mathcal{G}^f \varphi \equiv \lnot \mathcal{F}^f \lnot \varphi$ for $f \in \{\globalsucc,\globalpred,\abstractsucc,\caller,\parent,\child\}$.
For $f = \globalsucc$, we sometimes omit the superscript and write $\mathcal{F} \varphi$ etc.
Intuitively, these modalities correspond to the usual LTL modalities evaluated on the path starting in the current position and taking $f$-successors.
Additionally, we introduce modalities $\mathcal{F}^\reachable \varphi \equiv \mu X. (\varphi \lor \bigcirc^\globalsucc X \lor \bigcirc^\child X)$ and $\mathcal{G}^\reachable \varphi \equiv \lnot \mathcal{F}^\reachable \lnot \varphi$ to express that $\varphi$ holds in some position or all positions, respectively, reachable from the current position.
Using these abbreviations, dual operators and the equivalence $\lnot \lnot \varphi\equiv \varphi$ we can transform every well-formed formula into an equivalent formula in which negation only appears in front of atomic propositions.
We call this form \emph{positive normal form} and assume formulae to be given in this form in the algorithms presented in this paper.

\textbf{Semantics.}
We now formally define the semantics of \logicname.
It is defined with respect to an execution graph $G=(V, l, (\rightarrow^d)_{d\in\moves}, \curvearrowright)$ and a fixpoint variable assignment 
$\mathcal{V}$ assigning sets of nodes of $G$ to fixpoint variables.
Intuitively, $\llbracket \varphi \rrbracket^G_{\mathcal{V}}$ is the set of nodes of $G$ satisfying $\varphi$ when each free fixpoint variable $X$ is interpreted to hold at nodes $\mathcal{V}(X)$.
In the following, for a fixpoint variable assignment $\mathcal{V}$, a fixpoint variable $X$ and a set of nodes $M\subseteq V$, we write $\mathcal{V}[X \mapsto M]$ for the fixpoint variable assignment with $\mathcal{V}[X \mapsto M](X)=M$ and $\mathcal{V}[X \mapsto M](Y)=\mathcal{V}(Y)$ for all variables $Y\neq X$.
\begin{definition}[Semantics of \logicname]
	Let $G=(V, l, (\rightarrow^d)_{d\in\moves}, \curvearrowright)$ be an execution graph and $\mathcal{V}$ be a fixpoint variable assignment. The semantics of an \logicname\ formula with respect to $G$ and $\mathcal{V}$ is defined by
	\begin{align*}
		\llbracket \ap \rrbracket_\mathcal{V}^G :=&\quad \{x \in V \mid \ap \in l(x)\}\\
		\llbracket \lnot \varphi \rrbracket_\mathcal{V}^G :=&\quad V \setminus \llbracket \varphi \rrbracket_\mathcal{V}^G \\
		\llbracket \varphi_1 \lor \varphi_2 \rrbracket_\mathcal{V}^G :=&\quad \llbracket \varphi_1 \rrbracket_\mathcal{V}^G \cup \llbracket \varphi_2 \rrbracket_\mathcal{V}^G \\
		\llbracket X \rrbracket_\mathcal{V}^G :=&\quad \mathcal{V}(X) \\
		\llbracket \bigcirc^f \varphi \rrbracket_\mathcal{V}^G :=&\quad \{ x \in V \mid \Succ^G_f(x) \text{ is defined and } \Succ^G_f(x) \in \llbracket \varphi \rrbracket_\mathcal{V}^G \} \\
		\llbracket \mu X. \varphi \rrbracket_\mathcal{V}^G :=&\quad \bigcap \{ M \subseteq V \mid \llbracket \varphi \rrbracket_{\mathcal{V}[X \mapsto M]}^G \subseteq M \}
	\end{align*}
	where $\ap\in \AP$ is an atomic proposition, $X$ is a fixpoint variable and $f\in \{\globalsucc,\globalpred,\abstractsucc,\caller,\parent,\child\}$ is a successor type.
\end{definition}
In this semantics definition, two remarks are in order.
First, it is easy to see using Knaster-Tarski's fixpoint theorem \cite{Tarski1955} that for formulae $\varphi$ in positive normal form, $\llbracket \mu X. \varphi \rrbracket_\mathcal{V}^G$ characterises the least fixpoint of the monotone function $\alpha_S \colon 2^V \to 2^V$ with $\alpha_S(M) = \llbracket \varphi \rrbracket^G_{\mathcal{V}[X \mapsto M]}$ for $S = (G,\mathcal{V},X,\varphi)$.
Second, for closed \logicname\ formulae $\varphi$, the semantics does not depend on the fixpoint variable assignment.
For such formulae, we introduce additional semantic notations.
We write $\llbracket \varphi \rrbracket^G$ for $\llbracket \varphi \rrbracket_\mathcal{V}^G$ where $\mathcal{V}$ is an arbitrary fixpoint variable assignment and set $\llbracket \varphi \rrbracket:=\{G\in\ExGraphs\mid v_0\in \llbracket \varphi\rrbracket^G\}$.
For an execution graph $G$, we write $G \models \varphi$ for $G \in \llbracket \varphi \rrbracket$.
Finally, for a DPN $\mathcal{M}$, we write $\mathcal{M}\models \varphi$, iff $G \models \varphi$ for all $G\in \llbracket\mathcal{M}\rrbracket$.

In this paper, we consider the following decision problems for \logicname:
\begin{itemize}
	\item \emph{Model Checking:} Given a DPN $\mathcal{M}$ and a closed well-formed \logicname\ formula $\varphi$, does $\mathcal{M} \models \varphi$ hold?
	\item \emph{DPN Satisfiability:} Given a closed well-formed \logicname\ formula $\varphi$, is there a DPN $\mathcal{M}$ such that $\mathcal{M} \models \varphi$?
	\item \emph{Graph Satisfiability:} Given a closed well-formed \logicname\ formula $\varphi$, is there an execution graph $G$ such that $G \models \varphi$?
\end{itemize}
	\section{Example properties}\label{sec:examples}

We motivate the introduction of our new logic with some examples.

\textbf{Locking policies.}
In programming languages like Java, mutual exclusion between different threads on certain procedures or code blocks is realised via \textit{synchronized} procedures or blocks.
Internally, this feature works by acquiring a lock upon entering a synchronized procedure or block that is released when leaving the synchronized part of the code \cite{JavaSynchronized}.
Locks are thus acquired and released in a nested manner.
In DPNs, this synchronization mechanism can be modelled by including symbols for locks in the stack alphabet that are pushed onto the stack when acquiring a lock and removed from the stack when releasing it.
A call or return of a synchronized procedure is then modelled by taking two $\call$- or $\ret$-transitions of the DPN, respectively, one for pushing or popping the lock symbol and another one as usual.
We also include the lock symbols as atomic propositions that are assigned to corresponding configuration heads.
In this setup, the formula $\varphi_{l} := \mathcal{F}^\caller l$ expresses that the lock $l$ is currently held using the caller modality $\mathcal{F}^\caller$.
This form of modelling also works for reentrant locks, i.e. locks that can be acquired multiple times.
When threads acquire multiple locks, problems with deadlocks can occur when different threads acquire locks in a different order.
Assume, for example, that we have two locks where thread one acquires lock one first and then lock two and thread two acquires lock two first and then lock one.
In this case, a deadlock can occur when the threads are scheduled such that thread one acquires lock one and thread two acquires lock two.
A common policy to avoid deadlocks is to ensure that all threads acquire locks in the same order.
The formula $\varphi_{ij} := \mathcal{F}^\caller (l_i \land \mathcal{G}^\caller \lnot l_j)$ expresses that lock $l_i$ is currently held and when it was acquired, lock $l_j$ was not held.
It can be used in the formula $\mathcal{G}^\reachable (\varphi_{l_i} \land \varphi_{l_j}) \rightarrow \varphi_{ij}$ to express that $l_i$ is always acquired before $l_j$, if both locks are held.
The disjunction $\mathcal{G}^\reachable (\varphi_{l_i} \land \varphi_{l_j}) \rightarrow \varphi_{ij} \lor \mathcal{G}^\reachable (\varphi_{l_i} \land \varphi_{l_j}) \rightarrow \varphi_{ji}$ then expresses the existence of a global order for locks $l_i$ and $l_j$ and the existence of a global order for all locks can be expressed by a boolean combination of a quadratic number of variants of this formula.
Another problem with locking arises when certain threads wait for a lock that is held by another thread for an infinite amount of time, e.g. if a synchronized method is never left.
A policy addressing this problem is to ensure that all locks that are acquired are released at some point in the future.
We can express this using the formula $\mathcal{G}^\reachable \bigwedge_{l \in \mathit{Locks}} \varphi_l \rightarrow \mathcal{F} \lnot \varphi_l$.
Under these two policies, a necessary and sufficient condition for mutual exclusion of two program points labelled $s_1$ and $s_2$ is that a common lock is held at the two program points.
This can also be expressed in a formula from our logic: $\bigvee_{l \in \mathit{Locks}} \mathcal{G}^\reachable ((s_1 \rightarrow \varphi_l) \land (s_2 \rightarrow \varphi_l))$.

\textbf{Behaviour of main and worker threads.}
We elaborate on a motivating example for single indexed LTL from \cite{Song2015} expressible in \logicname.
In this example, a main thread of a server processes requests from clients by starting a worker thread responding to the specific request.
Then, the main thread should repeatedly accept new requests, expressed by the formula $\mathcal{G}^\reachable (\mathit{main} \rightarrow \mathcal{G}\mathcal{F} \mathit{accept})$.
Also, each worker thread should respond with a correct acknowledgement to each type of request, i.e. it should respond exactly with $\mathit{ack}$ to $\mathit{req}$ and exactly with $\mathit{ack'}$ to $\mathit{req'}$.
This is expressed by the formula $\mathcal{G}^\reachable (\mathit{worker} \rightarrow (\mathit{req} \rightarrow (\mathcal{F} \mathit{ack} \land \mathcal{G}\lnot \mathit{ack'}) \land \mathit{req'} \rightarrow (\mathcal{F} \mathit{ack'} \land \mathcal{G}\lnot \mathit{ack})))$.
Such requirements were already expressible in single indexed LTL.
However, using the different types of successor operators in \logicname, we can further expand on this scenario and express properties not expressible in single indexed LTL.
For example, it is a reasonable requirement that worker threads are only spawned by the main thread and only if the main thread has accepted a request.
This requirement can be expressed in the formula $\mathcal{G}^\reachable (\mathit{worker} \rightarrow \bigcirc^\parent (\mathit{main} \land \mathit{accept}))$.
Another desirable property in this scenario is a variant of the property from the introduction.
In particular, we may want worker threads to only be spawned from a procedure $\mathit{pr}$ which performs bookkeeping about the currently active worker threads.
This is expressed by the formula $\mathcal{G}^\reachable (\bigcirc^c \mathit{worker} \rightarrow \mathcal{F}^\caller \mathit{pr})$.

\textbf{Single indexed LTL model checking.}
It is no surprise that the previous motivating example for single indexed LTL is expressible in \logicname.
Indeed, we show that the full approach of single indexed LTL DPN model checking from \cite{Song2015} can also be handled using our logic.
We first sketch their setup.
In \cite{Song2015}, a DPN $\mathcal{M} = \{\mathcal{P}_1,\dots,\mathcal{P}_n\}$ is defined as a set of pushdown systems $\mathcal{P}_i$ with the ability to spawn threads executing one of the pushdown systems of $\mathcal{M}$.
A single indexed LTL formula is a conjunction $\varphi = \bigwedge_{i=1}^n \varphi_i$ of LTL formulae $\varphi_i$ that are each assigned to a specific pushdown system $\mathcal{P}_i$.
Then, $\mathcal{M} \models \varphi$ holds iff $\mathcal{M}$ has a \textit{global run} such that for all $i$, every \textit{local run} of $\mathcal{P}_i$ in the global run satisfies $\varphi_i$.
In our setup, their global runs correspond to execution graphs and their local runs correspond to the paths in the execution graph starting in positions where new threads are spawned and following the global successors.
Since in single indexed LTL model checking, the existence of a global run is checked, whereas in \logicname\ model checking, it is checked that all execution graphs satisfy a property, we can check that $\mathcal{M} \not\models \varphi$ for a single indexed LTL formula $\varphi = \bigwedge_{i=1}^n \varphi_i$ using \logicname\ model checking.
This is done as follows.
We model the partition of a DPN $\mathcal{M}$ from their setup into its pushdown systems $\mathcal{P}_i$ by labelling every control location of $\mathcal{P}_i$ with a fresh atomic proposition $p_i$ in its translation $\bar{\mathcal{M}}$ in our setup.
LTL formulae $\varphi_i$ can trivially be translated to \logicname\ by encoding until operators using least fixpoints.
Then, the \logicname\ formula $\bar{\varphi} = (p_1 \land \lnot \varphi_1) \lor \mathcal{F}^\reachable (\bigvee_{i=1}^n \bigcirc^{\child} (p_i \land \lnot \varphi_i))$ expresses that there is a local run of $\mathcal{P}_i$ for some $i$ that does not satisfy $\varphi_i$.
In this formula, the disjunct $(p_1 \land \lnot \varphi_1)$ identifies a violation by the root process $\mathcal{P}_1$ and the disjunct $\mathcal{F}^\reachable (\bigvee_{i=1}^n \bigcirc^{\child} (p_i \land \lnot \varphi_i))$ identifies violations by spawned processes.
Accordingly, $\mathcal{M} \models \varphi$ (in the single indexed LTL setup) iff $\bar{\mathcal{M}} \not\models \bar{\varphi}$ (in our setup).
	\section{From Graph Semantics to Tree Semantics}\label{sec:treesemantics}

In order to enable algorithmic verification with tree automata, we introduce an additional structure called \textit{execution tree}.
In a nutshell, these trees are obtained from execution graphs by keeping the same set of nodes and adjusting the edge relation a little.
In particular, we discard $\ret$-edges.
In order to properly interpret left and right children in this adjusted structure, we add labels $(l,d,p)$ where $l$ represents the label of the current node, $d$ represents the transition types from this node to its children and $p$ represents the transition type from the parent to this node.
This yields us a structure simpler than execution graphs that still contains the same information and can be analysed using tree automata.

\textbf{Execution trees.}
Let $G=(V, l, (\rightarrow^d)_{d\in\moves}, \curvearrowright)$ be an execution graph. 
We inductively define a map $\delta_G \colon V\to\{0,1\}^*$ assigning a tree node to each graph node $x\in V$ as follows.
\begin{itemize}
	\item If $x=v_0$, we set $\delta_G(x):=\varepsilon$.
	\item If $x$ has a $\call$- or $\intern$-predecessor $y$, we set $\delta_G(x):=\delta_G(y)\cdot 0$. In this case, we also call $\delta_G(x)$ a $\call$- or $\intern$-\textit{child} of $\delta_G(y)$, respectively. 
	\item If there is $y\in V$ such that $y$ is a $\spawn$-predecessor of $x$ or $y\curvearrowright x$, we set $\delta_G(x):=\delta_G(y)\cdot 1$. If $y$ is a $\spawn$-predecessor of $x$, we also call $\delta_G(x)$ a $\spawn$-\textit{child} of $\delta_G(y)$ and if $y\curvearrowright x$, we also call $\delta_G(x)$ a $\ret$-\textit{child} of $\delta_G(y)$.
\end{itemize}
Additionally, for a subset $M\subseteq \moves$ and nodes $x,y\in V$, we call $\delta_G(y)$ an $M$-\textit{descendant} of $\delta_G(x)$ and $\delta_G(x)$ an $M$-\textit{ancestor} of $\delta_G(y)$, if there is a path from $\delta_G(x)$ to $\delta_G(y)$ in the tree following only $M$-children.

Let $\treelabels = 2^\AP\times \{\intern,\call,\callret,\spawn,\ret,\labelend\} \times (\moves\cup\{\bot\})$ be the set of labels for tree nodes
and the arity function $\ar\colon\treelabels\to \{0,1,2\}$ be defined by $\ar(l,\ret,p) = \ar (l,\labelend,p) = 0$, $\ar(l,\intern,p) = \ar (l,\call,p) = 1$ and $\ar(l,\callret,p) = \ar (l,\spawn,p) = 2$.
The \textit{tree representation} $\mathcal{T}(G)$ of $G$ is the $(\treelabels,\ar)$-labelled binary tree $(\mathsf{im}(\delta_G),r)$ where $\mathsf{im}(\delta_G)=\{\delta_G(x)\mid x\in V\}$ denotes the image of $\delta_G$ and for all $x \in V$ we have $r(\delta_G(x)) = (l(x), d(x), p(x))$ where (i) either $p(x)\neq \bot$ and $x$ has a $p(x)$-predecessor or $p(x)=\bot$ and $x=v_0$ and (ii) one of the following conditions hold:
\begin{itemize}
	\item $x$ has only one $d(x)$-successor and $d(x)\in\{\intern,\ret\}$.
	\item $x$ has only one $\intern$- and one $\spawn$-successor and $d(x)=\spawn$.
	\item $x$ has only one $\call$-successor, there is no $y\in V$ with $x\curvearrowright y$, and $d(x)=\call$.
	\item $x$ has only one $\call$-successor, there is $y\in V$ with $x\curvearrowright y$, and $d(x)=\callret$.
	\item $x$ has no successors and $d(x)=\labelend$.
\end{itemize}
A tree representation of an execution graph is also called an \textit{execution tree}. 
An example of an execution tree can be found in \cref{fig:executiontree}.
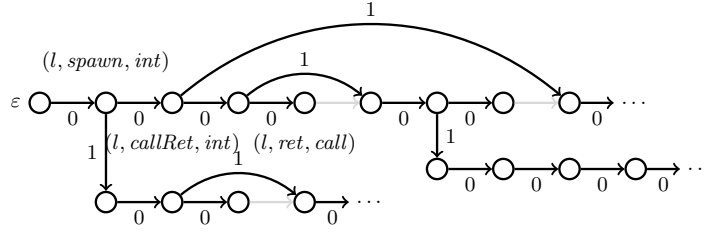
\begin{figure}
	\centering
	\resizebox{.8\textwidth}{!}{
	\begin{tikzpicture}
	\node[draw,circle,line width=1pt,label=left:$\varepsilon$] at (0,0) (n11) {};
	\node[draw,circle,line width=1pt,label={[label distance=5pt]\footnotesize$(l,\spawn,\intern)$}] at (1,0) (n12) {};
	\node[draw,circle,line width=1pt,label={[label distance=5pt]below:\footnotesize$(l,\callret,\intern)$}] at (2,0) (n13) {};
	\node[draw,circle,line width=1pt] at (3,0) (n14) {};
	\node[draw,circle,line width=1pt,label={[label distance=5pt]below:\footnotesize$(l,\ret,\call)$}] at (4,0) (n15) {};
	\node[draw,circle,line width=1pt] at (5,0) (n16) {};
	\node[draw,circle,line width=1pt] at (6,0) (n17) {};
	\node[draw,circle,line width=1pt] at (7,0) (n18) {};
	\node[draw,circle,line width=1pt] at (8,0) (n19) {};
	\node[] at (9,0) (dots1) {\dots};
	
	\node[draw,circle,line width=1pt] at (1,-1.5) (n21) {};
	\node[draw,circle,line width=1pt] at (2,-1.5) (n22) {};
	\node[draw,circle,line width=1pt] at (3,-1.5) (n23) {};
	\node[draw,circle,line width=1pt] at (4,-1.5) (n24) {};
	\node[] at (5,-1.5) (dots2) {\dots};
	
	\node[draw,circle,line width=1pt] at (6,-1) (n31) {};
	\node[draw,circle,line width=1pt] at (7,-1) (n32) {};
	\node[draw,circle,line width=1pt] at (8,-1) (n33) {};
	\node[draw,circle,line width=1pt] at (9,-1) (n34) {};
	\node[] at (10,-1) (dots3) {\dots};
	
	\path[draw,->,line width=1pt] (n11) to node[midway,below]{\footnotesize$0$} (n12);
	\path[draw,->,line width=1pt] (n12) to node[midway,below]{\footnotesize$0$} (n13);
	\path[draw,->,line width=1pt] (n13) to node[midway,below]{\footnotesize$0$} (n14);
	\path[draw,->,line width=1pt] (n14) to node[midway,below]{\footnotesize$0$} (n15);
	\path[draw,gray!30,->,line width=1pt] (n15) to (n16);
	\path[draw,->,line width=1pt] (n16) to node[midway,below]{\footnotesize$0$} (n17);
	\path[draw,->,line width=1pt] (n17) to node[midway,below]{\footnotesize$0$} (n18);
	\path[draw,gray!30,->,line width=1pt] (n18) to (n19);
	\path[draw,->,line width=1pt] (n19) to node[midway,below]{\footnotesize$0$} (dots1);

	\path[draw,->,out=40,in=140,line width=1pt] (n14) to node[midway,above]{\footnotesize$1$} (n16);
	\path[draw,->,out=40,in=140,line width=1pt] (n13) to node[midway,above]{\footnotesize$1$} (n19);
	
	\path[draw,->,line width=1pt] (n12) to node[midway,left]{\footnotesize$1$} (n21);
	
	\path[draw,->,line width=1pt] (n21) to node[midway,below]{\footnotesize$0$} (n22);
	\path[draw,->,line width=1pt] (n22) to node[midway,below]{\footnotesize$0$} (n23);
	\path[draw,gray!30,->,line width=1pt] (n23) to (n24);
	\path[draw,->,line width=1pt] (n24) to node[midway,below]{\footnotesize$0$} (dots2);
	
	\path[draw,->,out=40,in=140,line width=1pt] (n22) to node[midway,above]{\footnotesize$1$} (n24);
	
	\path[draw,->,line width=1pt] (n17) to node[midway,right]{\footnotesize$1$} (n31);
	
	\path[draw,->,line width=1pt] (n31) to node[midway,below]{\footnotesize$0$} (n32);
	\path[draw,->,line width=1pt] (n32) to node[midway,below]{\footnotesize$0$} (n33);
	\path[draw,->,line width=1pt] (n33) to node[midway,below]{\footnotesize$0$} (n34);
	\path[draw,->,line width=1pt] (n34) to node[midway,below]{\footnotesize$0$} (dots3);
	\end{tikzpicture}
	}
	\caption{Execution tree for the execution graph in \cref{fig:executiongraph}.
		An edge from node $t$ to node $t'$ labelled $d$ means that $t' = t \cdot d$. 
		Labels are depicted for nodes $0$, $00$ and $0000$.
		Gray edges exist in the execution graph but not in the execution tree.}
	\label{fig:executiontree}
\end{figure}

\textbf{Adapted successor functions.}
We adapt the successor functions previously defined on execution graphs to execution trees in order to allow us to check the satisfaction of formulae directly on execution trees.
Specifically, we define multiple successor functions $\Succ_\globalsucc^{\mathcal{T}}$, $\Succ_\globalpred^{\mathcal{T}}$, $\Succ_\abstractsucc^{\mathcal{T}}$, $\Succ_\caller^{\mathcal{T}}$, $\Succ_\parent^{\mathcal{T}}$ and $\Succ_\child^{\mathcal{T}}$ with signature $T \rightsquigarrow T$ for execution trees $\mathcal{T}=(T,r)$.
For $t\in T$ with $r(t)=(l,d,p)$, the successor functions are given as follows:
	\begin{itemize}
	
		\item The \emph{abstract successor} $\Succ_\abstractsucc^{\mathcal{T}}(t)$ of $t$ is defined as the left child of $t$, if $d\in\{\intern,\spawn\}$, the right child of $t$, if $d=\callret$, and undefined else.
		
		\item The \emph{caller predecessor} $\Succ_\caller^{\mathcal{T}}(t)$ of $t$ is defined as the parent node of $t$, if $p=\call$, the caller predecessor of its parent node, if $p\in\{\intern,\ret\}$ and this is defined, and undefined else.
		
		\item The \emph{global successor} $\Succ_\globalsucc^{\mathcal{T}}(t)$ of $t$ is defined as the left child of $t$, if $d\in\{\intern,\call,\callret,\spawn\}$, $\Succ_\abstractsucc^{\mathcal{T}}(\Succ_\caller^{\mathcal{T}}(t))$, if $d=\ret$, and undefined else.
		
		\item The \emph{global predecessor} $\Succ_\globalpred^{\mathcal{T}}(t)$ of $t$ is defined as the parent node of $t$, if $p\in \{\intern,\call\}$, the $\{\intern,\ret\}$-descendant leaf of the left child of its parent node, if $p=\ret$, and undefined else.
		
		\item The \emph{parent predecessor} $\Succ_\parent^{\mathcal{T}}(t)$ of $t$ is defined as the parent node of $t$, if $p\!=\!\spawn$, the parent predecessor of its parent node, if $p\in\{\intern,\call,\ret\}$ and this is defined, and undefined else.
		
		\item The \emph{child successor} $\Succ_\child^{\mathcal{T}}(t)$ of $t$ is defined as the right child of $t$, if $d=\spawn$, and undefined else.
	\end{itemize}

We show in the following lemma that these adapted successor functions behave exactly like their counterparts on execution graphs.

\begin{lemma}\label{lemma:successorequivalence}
	Let $G=(V, l, (\rightarrow^d)_{d\in\moves},\curvearrowright)$ be an execution graph with $\mathcal{T}(G)=\mathcal{T}$. For all $f\in\{\globalsucc,\globalpred,\abstractsucc,\caller,\parent,\child\}$ we have $\delta_G\circ\Succ^G_f=\Succ_f^{\mathcal{T}}\circ \delta_G$, i.e. for all nodes $x \in V$, $\delta_G(\Succ^G_f(x))$ is defined iff $\Succ_f^{\mathcal{T}}(\delta_G(x))$ is defined and in this case $\delta_G(\Succ^G_f(x))=\Succ_f^{\mathcal{T}}(\delta_G(x))$.
\end{lemma}
A detailed proof of this lemma can be found in \cref{app:treesemanticsproofs}.
	\section{Model Checking and Satisfiability}\label{sec:decisionproblems}

We now use execution trees to decide the model checking and satisfiability problems for \logicname.
For this, we construct three tree automata: one automaton for checking whether a tree is an execution tree, a second automaton for checking whether an execution graph (given by its tree representation) satisfies a given formula, and another automaton for checking whether a tree represents an execution graph generated by a given DPN.

\textbf{An automaton for execution trees.}
We first construct a nondeterministic parity tree automaton that checks whether a $(\treelabels,\ar)$-labelled binary tree is an execution tree. At each node labelled by $(l,d,p)$, the automaton needs to ensure that the node is a $p$-child, if $p\neq \bot$, and that it is the root, if $p=\bot$.
Moreover, if $d=\callret$, it has to check that its $\call$-child does have an $\{\intern,\ret\}$-descendant leaf.
Finally, it has to ensure that for leaves $t$ labelled by $(l,d,p)$ we have $d=\ret$ iff $t$ is the $\{\intern,\ret\}$-descendant leaf of a $\call$-child of a node labelled by $(l',\callret,p')$ for some $l'\in 2^\AP$ and $p'\in\moves\cup\{\bot\}$.
Thus, we can define the automaton as $\mathcal{A}_\ExTrees=(Q,q_0,\rho,\Omega)$ with state set $Q=(\moves\cup\{\bot\})\times \{0,1\}$ and initial state $q_0=(\bot,0)$.
Intuitively, in a state $(p,c)$, $p$ denotes the parent edge type and the bit $c$ indicates whether the current node is an $\{\intern,\ret\}$-descendant of a $\call$-child of a node labelled by $(l',\callret,p')$ for some $l'\in 2^\AP$ and $p'\in\moves\cup\{\bot\}$.
The transition function $\rho$ is defined by
\begin{align*}
\rho((p,c),(l,d,p')):=\begin{cases}
(0,(\intern,c)) & \text{ if } d=\intern\\
(0,(\call,0)) & \text{ if } d=\call \text{ and } c=0\\
(0,(\call,1)) \land (1,(\ret,c)) & \text{ if } d=\callret\\
(0,(\intern,c)) \land (1,(\spawn,0)) & \text{ if } d=\spawn\\
\true & \text{ if } (d,c)\in\{(\ret,1),(\labelend,0)\}
\end{cases}
\end{align*}
for $p=p'$ and $\rho((p,c),(l,d,p')):=\false$ in all other cases.
The priority assignment is given by $\Omega(p,c)=c$ for all $(p,c)\in Q$.

We establish the following theorem.
A proof can be found in \cref{app:decisionproblemsproofs}.
\begin{theorem}\label{thm:extreeautomaton}
One can construct a NPTA $\mathcal{A}_\ExTrees$ over $(\treelabels,\ar)$-labelled binary trees with a constant size such that $\mathcal{L}(\mathcal{A}_\ExTrees) = \{ \mathcal{T}(G) \mid G \text{ is an execution graph}\}$.
\end{theorem}

\textbf{An automaton for formulae.}
For the next automaton, we define a 2-way alternating tree automaton evaluating $\varphi$ on execution trees, intersect it with the automaton recognising execution trees and then transform this automaton into a nondeterministic parity tree automaton. 
In the following, let $\varphi$ be a closed, well-formed \logicname\ formula in positive normal form. 
We define the automaton for $\varphi$ as $\tilde{\mathcal{A}}_\varphi = (Q,q_0,\rho,\Omega)$ where $Q$, $q_0$, $\rho$ and $\Omega$ are described in more detail in the following paragraphs.

The states of the automaton are given by 
\begin{align*}
	Q=&\ \Sub(\varphi)\cup Q_1\cup Q_2 \text{ where } \\
	Q_1=&\ \{\bigcirc^\caller\bigcirc^\abstractsucc\psi,\bigcirc^\abstractsucc\psi\mid\bigcirc^\globalsucc\psi\in\Sub(\varphi)\text{ or } \bigcirc^{\overline{\globalsucc}}\psi\in\Sub(\varphi)\} \text{ and }\\
	Q_2=&\ \{\call,\leaf\}\times \{\psi\mid\bigcirc^\globalpred\psi\in\Sub(\varphi)\text{ or } \bigcirc^{\overline{\globalpred}}\psi\in\Sub(\varphi)\}
\end{align*}
with initial state $q_0 = \varphi$.
Since we use another automaton to check that the given tree indeed represents an execution graph, we care only about execution trees as inputs in this construction. 
Intuitively, being in a state $\psi \in \Sub(\varphi)\cup Q_1$ at the position $\delta_G(x)$ in the input execution tree $\mathcal{T}(G)$, the automaton checks whether the node $x$ satisfies $\psi$, i.e. whether $x \in \llbracket \psi \rrbracket^G$.
The states in $Q_2$ are used to handle the global predecessor next modality and its dual version.
We use states of the form $(\call,\psi)$ to denote that we should move to the $\call$-child of the current node and switch to state $(\leaf,\psi)$; states of the form $(\leaf, \psi)$ denote that we should check $\psi$ for the $\{\intern,\ret\}$-descendant leaf of the current node.

The transition function $\rho$ is defined as described next.
Recall that $\tilde{\mathcal{A}}_\varphi$ operates on execution trees which are labelled by triples $(l,d,p)$ where $l\in 2^\AP$ are the atomic propositions, $d\in\{\intern,\call,\callret,\spawn,\ret,\labelend\}$ specifies the successor types of the current node and $p\in\moves\cup\{\bot\}$ denotes the type of its predecessor.
If the current state is an atomic proposition or a negation of an atomic proposition, we can check directly whether the tree node is labelled by this proposition and thus determine whether the formula holds:
\begin{align*}
	\rho(\ap,(l,d,p)) := 
	\begin{cases}
		\true \text{ if } \ap\in l\\
		\false \text{ if } \ap\notin l
	\end{cases}
	&&\rho(\lnot \ap,(l,d,p)) := 
	\begin{cases}
		\false \text{ if } \ap\in l\\
		\true \text{ if } \ap\notin l.
	\end{cases}
\end{align*}

For a disjunction or conjunction of two formulae, we can use the power of alternation and set 
\begin{align*}
	\rho(\psi_1\lor \psi_2,\sigma) := (\varepsilon, \psi_1)\lor (\varepsilon,\psi_2) \text{ and } \rho(\psi_1\land \psi_2, \sigma) := (\varepsilon, \psi_1)\land (\varepsilon,\psi_2).
\end{align*}

For a formula of the form $\bigcirc^f\psi$, we move to the corresponding successor of the given node and then switch to state $\psi$.
In most cases, the according transitions can be defined straightforwardly using the characterisation from the successor functions on execution trees:
\begin{align*}
	&\rho(\bigcirc^\globalsucc\psi, (l,d,p))
	:= \begin{cases}
	(0, \psi) \text{ if } d\in\{\intern,\call,\callret,\spawn\}\\
	(\varepsilon, \bigcirc^{\caller}\bigcirc^{\abstractsucc}\psi) \text{ if } d=\ret\\
	\false  \text{ if } d=\labelend
	\end{cases}
\end{align*}
\begin{align*}
	&\rho(\bigcirc^\abstractsucc\psi, (l,d,p)) && \rho(\bigcirc^\caller\psi, (l,d,p)) \\
		& := \begin{cases}
			(0, \psi) \text{ if } d\in\{\intern,\spawn\}\\
			(1, \psi) \text{ if } d=\callret\\
			\false \text{ if } d\in\{\call,\ret,\labelend\}
		\end{cases}
	 && := 
		\begin{cases}
			(\uparrow, \psi) \text{ if } p=\call\\
			(\uparrow, \bigcirc^\caller\psi) \text{ if } p\in\{\intern,\ret\}\\
			\false \text{ if } p\in\{\spawn,\bot\}
		\end{cases}\\
	& \rho(\bigcirc^\parent\psi, (l,d,p)) && \rho(\bigcirc^\child\psi, (l,d,p))\\
		& := \begin{cases}
			(\uparrow, \psi) \text{ if } p=\spawn\\
			(\uparrow, \bigcirc^\parent\psi) \text{ if } p\in\{\intern,\call,\ret\}\\
			\false \text{ if } p=\bot
		\end{cases}
	 && := 
		\begin{cases}
			(1, \psi) \text{ if } d=\spawn\\
			\false \text{ if } d\neq\spawn
		\end{cases}
\end{align*}
In the above definition, we move to $\false$ when we see that the desired successor does not exist and the formula is not satisfied.
The transition function for dual next operators is defined analogously but moves to $\true$ instead of $\false$ in case the successor does not exist.

For the global predecessor, we additionally use states of the form $(\call,\psi)$ and $(\mathit{leaf},\psi)$ for moving to the $\{\intern,\ret\}$-descendant leaf of the $\call$-child of the parent of a node in cases where the global predecessor is defined this way:
\begin{align*}
&\rho(\bigcirc^{\globalpred}\psi, (l,d,p)) && \rho((\mathit{leaf},\psi), (l,d,p))\\
		& := \begin{cases}
			(\uparrow, \psi) \text{ if } p\in \{\intern,\call\}\\
			(\uparrow, (\call,\psi)) \text{ if } p=\ret\\
			\false \text{ if } p\in\{\spawn,\bot\},
		\end{cases}
		&& := \begin{cases}
			(0, (\mathit{leaf},\psi)) \text{ if } d\in \{\intern,\spawn\}\\
			(1,(\mathit{leaf},\psi)) \text{ if } d=\callret\\
			(\varepsilon, \psi) \text{ if } d\in\{\ret,\call,\labelend\}\\
		\end{cases}\\
	&\text{and }\rho((\call,\psi), \sigma) := (0,(\mathit{leaf},\psi)).
\end{align*}
Note that if we are in a state $(\mathit{leaf},\psi)$ at position $\delta_G(x)$ in the tree, $d(x)\in\{\call,\labelend\}$ cannot hold if the tree represents an execution graph since in this case $x$ lies on the path between nodes $y$ and $z$ following $\moves\setminus\{\spawn\}$-successors with $y\curvearrowright z$ and $x\neq z$.

Finally, fixpoint formulae lead to loops:
\begin{align*}
	\rho(\lambda X.\psi,\sigma) := (\varepsilon, \psi) \text{ for } \lambda \in \{\mu,\nu\} \text{ and } \rho(X, \sigma) := (\varepsilon, \fp(X)).
\end{align*}

The acceptance condition specifies whether a fixpoint formula may be visited at most a finite number of times or an infinite number of visits is allowed.
In this definition, higher priorities are assigned to fixpoint formulae binding variables which \textit{depend} on other fixpoint variables.
Formally, we say that a fixpoint variable $X'$ depends on the variable $X$ in $\varphi$, written $X \prec_\varphi X'$, if $X$ is a free variable in $\fp(X')$.
We consider all maximal chains $X_1\prec_\varphi ... \prec_\varphi X_n$ of fixpoint variables appearing in $\varphi$.
If $\fp(X_1)$ is a formula of the form $\mu X.\psi$, we set $\Omega(\fp(X_1))=1$, otherwise we set $\Omega(\fp(X_1))=0$.
Then, we move through the chains and assign this priority to $\fp(X_i)$ as long as the fixpoint type does not change.
In that case, we increase the currently assigned priority by one and keep going.
Then, we set $\Omega(q)$ to the highest priority assigned so far for all other states $q$.

We establish the following theorem.
\begin{theorem}\label{thm:formulatoautomaton}
	Let $\varphi$ be a closed, well-formed \logicname\ formula, $G$ be an execution graph and $\tilde{\mathcal{A}}_\varphi$ be the 2ATA defined above.
	Then $\tilde{\mathcal{A}}_\varphi$ accepts $\mathcal{T}(G)$ iff $G\in\llbracket \varphi\rrbracket$.
\end{theorem}
\begin{proof}[Proof Sketch]
	The proof is by induction on the structure of $\varphi$.
	Therefore, we also have to deal with non-closed subformulae and consider valuations to decide whether a subformula is satisfied.
	In order to do this in a formal way, we consider automata with special states $X_1,\dots,X_n$, called \emph{holes} \cite{Lange2005}, that can be filled with sets of nodes $L_1,\dots,L_n$ of a given tree.
	Intuitively, such an automaton can operate on a tree as before, but when a hole $X_i$ is encountered during a run and we are at the tree node $t$, then we do not continue on the current path and say that it is accepting iff $t\in L_i$.
	By $\mathcal{L}_{q}^\mathcal{T}(\mathcal{A}[X_1:L_1,\dots,X_n:L_n])$ we denote the set of nodes $t\in T$ such that there is an accepting $(t,q)$-run over $\mathcal{A}$ where the states $X_1,\dots,X_n$ are holes filled by $L_1,\dots,L_n$.
	
	For the inductive proof, we assume that the free variables of the current formula $\psi\in\Sub(\varphi)$ are holes in the automaton and show that the language of this automaton corresponds to the semantics of $\psi$. 
	Intuitively, we fill the holes in the automaton, i.e. the free variables of $\psi$, with the same sets of nodes as specified by a given valuation that we consider for the semantics of $\psi$. 
	More formally, the holes are filled by sets of tree nodes that correspond to given sets of graph nodes in the valuation.
	
	We consider the case for subformulae of the form $\psi\equiv \mu X.\psi'$ with free variables $X_1,\dots, X_n$. 
	Let $\mathcal{V}$ be a fixpoint variable assignment, $\mathcal{T}(G)=\mathcal{T}=(T,r)$ and
	$R$ be a $(t,\psi)$-run over $\mathcal{\tilde{A}}_\varphi$ for a $t\in T$, where the states $X_1,\dots,X_n$ are holes filled by $\delta_G(L_1),\dots,\delta_G(L_n)$ with $L_i=\mathcal{V}(X_i)$. 
	We observe that $R$ can only visit states $\varphi'$ of the form $\mu X.\psi''$ or $\nu X.\psi''$ if $\varphi'$ is a subformula of $\psi$.
	Therefore, $\Omega(\psi)$ is the lowest priority occurring in the run so that the state $\psi$ can only be visited finitely often if the run is accepting.
	This means we can characterize $\mathcal{L}_\psi^{\mathcal{T}}(\mathcal{\tilde{A}}_\varphi[X_1:{\delta}_G(L_1),\dots,X_n:{\delta}_G(L_n)])$ as the least fixpoint of the function $f:2^T\to 2^T$ with $f({\delta}_G(L)):=\mathcal{L}_{\psi'}^{\mathcal{T}}(\mathcal{\tilde{A}}_\varphi[X_1:{\delta}_G(L_1),\dots,X_n:{\delta}_G(L_n), X:{\delta}_G(L)])$.
	Thus, we can use the induction hypothesis and the fixpoint characterization of the semantics of $\psi$ obtained by Knaster-Tarski's fixpoint theorem to get the desired result in this inductive step.
	
	Since $\varphi$ is closed, the induction establishes in particular that
	$\tilde{\mathcal{A}}_\varphi$ accepts $\mathcal{T}(G)$ iff $G\in\llbracket \varphi\rrbracket$.
	Details of this proof can be found in \cref{app:decisionproblemsproofs}.
	\hfill\qed
\end{proof}

As mentioned, we do not use this automaton directly but instead intersect it with $\mathcal{A}_\ExTrees$ from \cref{thm:extreeautomaton} and then transform it into a nondeterministic parity tree automaton using \cref{prop:treedealternation}.
We obtain:
\begin{corollary}\label{cor:formulatoautomaton}
	Let $\varphi$ be a closed, well-formed \logicname\ formula. 
	Then we can construct an NPTA $\mathcal{A}_\varphi$ over $(\treelabels,\ar)$-labelled binary trees with a number of states exponential and an acceptance condition linear in $|\varphi|$ such that $\mathcal{L}(\mathcal{A}_\varphi) = \{ \mathcal{T}(G) \mid G \text{ is an execution graph with } G\in\llbracket \varphi\rrbracket \}$.
\end{corollary}

\textbf{An automaton for DPNs.}
We proceed with an automaton for a DPN $\mathcal{M}=(S,s_0,\gamma_0,\Delta,L)$.
We define $\mathcal{A}_\mathcal{M}$ as an NPTA that checks whether an execution tree represents an execution graph generated by $\mathcal{M}$.
We set $\mathcal{A}_\mathcal{M}:=(Q,q_0,\rho,\Omega)$ where $Q$, $q_0$, $\rho$ and $\Omega$ are described in more detail next.

The state set is given by $Q=S\times \Gamma \times ((S\times \Gamma)\cup\{\bot\})$
with initial state $q_0=(s_0,\gamma_0,\bot)$. 
Being in a state $(s,\gamma,c)\in Q$ at the position $\delta_G(x)$ in the tree labelled by $(l,d,p)$ means that there is a suitable assignment $\assignment$ assigning configurations to the graph nodes whose corresponding tree nodes have been visited so far where $\assignment(x)=(s,\gamma w)$ for some stack content $w\in\Gamma^*\bot$.
If $d=\callret$, we also have to know the configuration assigned to the global predecessor of the $\ret$-child of the current node to check that we can extend $\assignment$ suitably for the children of the current node.
We thus guess this configuration in this case and use $c\in S\times \Gamma$ to indicate that we must assign $c$ to the $\{\intern,\ret\}$-descendant leaf of the call successor of the current node in order to fulfill the requirements for the assignment $\assignment$.
Note that the $\{\intern,\ret\}$-descendant leaf exists in this case, if the input tree is an execution tree.
The transition function $\rho$ then checks that
(i) $l=L(\assignment(x))$,
(ii) if $c\in S\times \Gamma$, then the configuration $c$ is assigned to the $\{\intern,\ret\}$-descendant leaf of $\delta_G(x)$ and
(iii) the assignment $\assignment$ can be properly extended to the children of $\delta_G(x)$.
We set
\begin{align*}
	\rho((s,\gamma,c),(l,\intern,p))  := & \bigvee\{(0, (s',\gamma',c))\mid s\gamma\rightarrow s'\gamma'\in\Delta_I\},\\
	\rho((s,\gamma,\bot),(l,\call,p))  := & \bigvee\{(0,(s',\gamma',\bot))\mid \exists \gamma''\in\Gamma \text{ s.t. } s\gamma\rightarrow s'\gamma'\gamma''\in\Delta_C\},\\
	\rho((s,\gamma,c),(l,\callret,p)) := & \bigvee\{(0,(s',\gamma',(s_r,\gamma_r)))\land (1, (s'',\gamma'',c))\mid\\
 	& s\gamma\rightarrow s'\gamma'\gamma''\in\Delta_C \text{ and } s_r\gamma_r\rightarrow s''\in\Delta_R\},\\
	\rho((s,\gamma,c),(l,\spawn,p)) := & \bigvee\{(0,(s',\gamma',c))\land (1, (s_n,\gamma_n,\bot)) \mid\\
	& s\gamma \rightarrow s'\gamma' \vartriangleright s_n \gamma_n \in \Delta_S\},\\
	\rho((s,\gamma,(s,\gamma)),(l,\ret,p)) := &\ \true \text{ and}\\
	\rho((s,\gamma,\bot),(l,\labelend,p)) := & \begin{cases}
	\true \text{ if there is no transition for } s\gamma \text{ in } \Delta\\
	\false \text{ else}
	\end{cases}
\end{align*}
for $l=L(s,\gamma)$ and $\rho((s,\gamma,c),(l,d,p)):=\false$ in all other cases.
Since we are only concerned with execution trees as inputs, all conditions necessary to determine if the input tree is generated by $\mathcal{M}$ are already checked by the transition function of $\mathcal{A}_\mathcal{M}$.
We thus set $\Omega(q):=0$ for all $q\in Q$.
We establish the following theorem. A detailed proof can be found in \cref{app:decisionproblemsproofs}.
\begin{theorem}\label{thm:dpnautomaton}
	Let	$\mathcal{M}$ be a DPN.
	We can construct an NPTA $\mathcal{A}_\mathcal{M}$ over $(\treelabels,\ar)$-labelled binary trees with a number of states quadratic in $|\mathcal{M}|$ and a trivial acceptance condition such that for all execution graphs $G$, $\mathcal{T}(G) \in \mathcal{L}(\mathcal{A}_\mathcal{M})$ iff $G \in \llbracket\mathcal{M}\rrbracket$.
\end{theorem}

\textbf{Complexity of Model Checking and Satisfiability.}
These automata constructions can be used to obtain a decision procedure for the model checking and satisfiability problems.
For the former, we obtain the following theorem:
\begin{theorem}
	The model checking problem for \logicname\ is $\EXPTIME$-complete.
	For fixed formulae, the problem is in $\PTIME$.
\end{theorem}
\begin{proof}
	For the upper bound, we construct an automaton for the negation of the formula using \cref{cor:formulatoautomaton} and intersect it with an automaton for the DPN from \cref{thm:dpnautomaton}.
	Since the acceptance condition of the latter is trivial, the resulting automaton is quadratic in the size of the DPN and exponential in the size of the formula by \cref{prop:treeintersection} (ii).
	It is tested for emptiness using \cref{prop:treeemptiness} in time exponential in $|\varphi|$ and polynomial in $|\mathcal{M}|$ to answer the model checking problem.
	
	The lower bound follows by a reduction from the $\mathit{LTL}$ pushdown model checking problem which was shown to be $\EXPTIME$-hard in \cite{Bouajjani1997}.
	The reduction is trivial since $\mathit{LTL}$ is a sublogic of \logicname\ for single threads and pushdown systems can be trivially embedded into DPNs with a single thread.
	\hfill\qed
\end{proof}
For satisfiability, we can show that the two problems defined in \cref{sec:logic} are equivalent and thus only need to solve one of the problems by a direct procedure.
\begin{theorem}\label{thm:satequivalence}
	The graph and DPN satisfiability problems are equivalent.
\end{theorem}
\begin{proof}
	For the first direction, assume that a formula $\varphi$ is satisfiable by a DPN $\mathcal{M}$.
	Then $G \models \varphi$ for all $G \in \llbracket \mathcal{M} \rrbracket$.
	Since $\llbracket \mathcal{M} \rrbracket \neq \emptyset$ (this indeed holds for all DPNs), we can thus choose an arbitrary graph $G \in \llbracket \mathcal{M} \rrbracket$ to show that $\varphi$ is satisfiable by a graph.
	
	For the other direction, assume that a formula $\varphi$ is satisfiable by a graph $G$.
	By \cref{cor:formulatoautomaton}, we know that $\mathcal{T}(G) \in \mathcal{L}(\mathcal{A}_\varphi)$.
	Since $\mathcal{L}(\mathcal{A}_\varphi)$ is a nonempty $\omega$-regular tree language, we know that $\mathcal{T} \in \mathcal{L}(\mathcal{A}_{\varphi})$ for a regular tree $\mathcal{T} = (T,r)$, i.e. a tree with only finitely many non-isomorphic subtrees (see e.g. Cor 8.20. in \cite{Graedel2002}).
	Let $x_1,\dots,x_n$ be the finitely many classes of nodes associated with the roots of the distinct subtrees of $T$ such that $x_1$ is the class of $\varepsilon$ and let $(l_i,d_i,p_i)$ be the label of the nodes from class $x_i$.
	We construct a DPN $\mathcal{M} = (\{s\},s,x_1,\Delta,L)$ with stack alphabet $\Gamma = \{x_1,\dots,x_n\}$.
	The labeling $L$ is defined such that $L(s,x_i) = l_i$.
	Transition rules are defined from the parent-child relationships between the different classes of nodes: (i) if $d_i = \intern$, then nodes of class $x_i$ have exactly one child of class $x_j$ and we include $s x_i \rightarrow s x_j \in \Delta$, (ii) if $d_i = \spawn$, then nodes of class $x_i$ have exactly one left child of class $x_j$ and one right child of class $x_k$ and we include $s x_i \rightarrow s x_j \triangleright s x_k \in \Delta$, (iii) if $d_i = \callret$, then nodes of class $x_i$ have exactly one left child of class $x_j$ and one right child of class $x_k$ and we include $s x_i \rightarrow s x_j x_k \in \Delta$, (iv) if $d_i = \call$, then nodes of class $x_i$ have exactly one child of class $x_j$ and we include $s x_i \rightarrow s x_j x_i \in \Delta$, (v) if $d_i = \ret$, then nodes of class $x_i$ have no children and we include $s x_i \rightarrow s \in \Delta$ and (vi) if $d_i = \labelend$, then nodes of class $x_i$ have no children and we do not include a transition.
	It is easy to see that $\llbracket \mathcal{M} \rrbracket$ is a singleton set since $\mathcal{M}$ is deterministic.
	We show that $\llbracket \mathcal{M} \rrbracket = \{ H \}$ where $\mathcal{T} = \mathcal{T}(H)$ and thus $\mathcal{M} \models \varphi$.
	For this, let $\mathcal{T}(H) = (T_H,r_H)$.
	
	We show by induction on the length of $x$ that for all $x \in \{0,1\}^*$, $x \in T$ iff $x \in T_H$ and in that case (a) $r(x) = r_H(x)$ and (b) if $x$ belongs to class $x_i$, then the configuration in $\delta_G^{-1}(x)$ is $(s,x_i w)$ for some stack content $w$.
	In the base case, we know that $\varepsilon \in T$ and $\varepsilon \in T_H$.
	We know that $r(\varepsilon) = (l_1,d_1,p_1)$ since $T$ is rooted in $x_1$ and $p_1 = \bot$ since $\mathcal{T}$ is an execution tree.
	Let $r_H(\varepsilon) = (l,d,p)$.
	Since $(s, x_1 \bot)$ is the starting configuration of $\mathcal{M}$, we know that it is also the configuration in $\delta_G^{-1}(\varepsilon)$ and that $l = l_1$.
	Additionally, we can show that $d = d_1$ by a case distinction on $d_1$.
	We only sketch the case $d_1 = \intern$, the other cases are similar.
	In this case, the only enabled transition in $(s, x_1 \bot)$ is $s x_1 \rightarrow s x_j$, an internal transition.
	Thus, $\delta_G^{-1}(\varepsilon)$ has exactly one $\intern$-successor in $H$ which means that $d = \intern$.
	Finally, since $\mathcal{T}(H)$ is an execution tree, we have $p = \bot$. 
	
	In the inductive step, we consider $x \cdot d$ for $d \in \{0,1\}$.
	From the induction hypothesis, we know that the claim holds for $x$.
	If $x \not\in T$ and $x \not\in T_H$, then also $x \cdot d \not\in T$ and $x \cdot d \not\in T_H$ since trees are prefix-closed.
	In the other case, let $x_i$ be the class of $x$.
	We have $x \in T$ and $x \in T_H$ with $r(x) = r_H(x) = (l_i,d_i,p_i)$ and the configuration in $\delta_G^{-1}(x)$ is $(s,x_i w)$ for some stack content $w$.
	We distinguish cases based on $d_i$.
	We consider the most involved case where $d_i = \callret$.
	Since $\mathcal{T}$ is an execution tree, we know that $x \cdot d \in T$ for $d \in \{0,1\}$.
	Let $x_j$ be the class of $x \cdot 0$ and $x_k$ be the class of $x \cdot 1$.
	We know that the only enabled transition in $(s,x_i w)$ is $s x_i \rightarrow s x_j x_k$.
	Since $d_i = \callret$ and since $\mathcal{T}(G)$ is an execution tree, we know that $\delta_G^{-1}(x \cdot 0)$ continues with the configuration after this call transition and $\delta_G^{-1}(x \cdot 1)$ continues with the configuration after the matching return transition (which exists in this case).
	Thus, the configuration in $\delta_G^{-1}(x \cdot 0)$ is $(s,x_j x_k w)$ and the configuration in $\delta_G^{-1}(x \cdot 1)$ is $(s, x_k w)$, establishing this part of the claim.
	We now establish that $r(x \cdot d) = r_H(x \cdot d)$.
	For the first and second component, this is established by the fact that the configuration in $\delta_G^{-1}(x \cdot d)$ determines both the label and the unique enabled transition.
	For the third component, this follows from the fact that both $\mathcal{T}$ and $\mathcal{T}(H)$ are execution trees and the fact that $r(x)$ and $r_H(x)$ match in the second component.
	\hfill\qed
\end{proof}
We obtain the following theorem for the two satisfiability problems:
\begin{theorem}\label{thm:sat}
	The graph and DPN satisfiability problems for \logicname\ are $\EXPTIME$-complete.
\end{theorem}
\begin{proof}
	Since the two problems are equivalent by \cref{thm:satequivalence}, we need to only give an upper and lower bound for the graph satisfiability problem.
	
	For the upper bound, we can construct an automaton for the formula using \cref{cor:formulatoautomaton} and test it for emptiness using \cref{prop:treeemptiness} in time exponential in $|\varphi|$ for an answer to the graph satisfiability problem.
	
	The lower bound follows by a reduction from the $\mathit{VP}$-$\mu$-$\mathit{TL}$ satisfiability problem which was shown to be $\EXPTIME$-hard in \cite{Bozzelli2007}.
	The reduction is straightforward since $\mathit{VP}$-$\mu$-$\mathit{TL}$ is a sublogic of \logicname\ and we can easily extract a nested word satisfying a formula interpreted in $\mathit{VP}$-$\mu$-$\mathit{TL}$ from the execution graph satisfying the same formula interpreted in \logicname.
	\hfill\qed
\end{proof}
	\section{Conclusion}\label{sec:conclusion}

We introduced a novel specification logic called \logicname\ for reasoning about the call-return and thread creation behaviour of dynamic pushdown networks.
We showed that a variety of interesting properties regarding the behaviour of multithreaded software is expressible in \logicname.
Further, the model checking and satisfiability problems were investigated.
The complexity of these problems is not higher than that of the corresponding problems for related logics for pushdown systems despite a more powerful logic and system model.
For future work, it would be interesting to consider more powerful variants of DPNs that allow communication and synchronization of different threads via locking or messages.
	\bibliographystyle{splncs04}
	\bibliography{sections/conclusion/citations}
	\clearpage
\appendix

\section{Fixpoint Theory}\label{app:fixpoints}

In some of the proofs in this paper, we need results from fixpoint theory.
We provide the necessary definitions and results in this section.
A \textit{partial order} is a pair $(L, \sqsubseteq)$ such that $\sqsubseteq$ is a reflexive, transitive and antisymmetric binary relation on $L$.
For $X \subseteq L$ and $x \in L$, we call $x$ a \textit{lower bound} on $X$ iff $x \sqsubseteq x'$ for all $x' \in X$.
Similarly, $x$ is called an \textit{upper bound} on $X$ iff $x' \sqsubseteq x$ for all $x' \in X$.
A lower bound $x$ of $X$ is called the \textit{greatest lower bound} of $X$, denoted $x = \bigsqcap X$, iff $x' \sqsubseteq x$ for all lower bounds $x'$ of $X$.
Analogously, an upper bound $x$ of $X$ is called the \textit{least upper bound} of $X$, denoted $x = \bigsqcup X$, iff $x \sqsubseteq x'$ for all upper bounds $x'$ of $X$.
A partial order $(L,\sqsubseteq)$ is called a \textit{complete lattice} iff the least upper bound $\bigsqcup X$ exists for every set $X \subseteq L$.
For a function $f \colon L \to L'$ on partial orders $(L,\sqsubseteq)$ and $(L',\sqsubseteq')$, we call $f$ \textit{monotone} iff $x \sqsubseteq x'$ implies $f(x) \sqsubseteq' f(x')$ for all $x,x' \in L$.
For $(L,\sqsubseteq)=(L',\sqsubseteq')$, a \textit{fixpoint} of $f$ is an element $x \in L$ with $f(x) = x$.
We call a fixpoint $x$ of $f$ the \textit{least fixpoint} of $f$, denoted $\mu f$, iff $x \sqsubseteq x'$ for all fixpoints $x'$ of $f$.
Analogously, a fixpoint $x$ is called the \textit{greatest fixpoint} of $f$, denoted $\nu f$, iff $x' \sqsubseteq x$ for all fixpoints $x'$ of $f$.
We use the classical Knaster-Tarski fixpoint theorem:
\begin{proposition}[\cite{Tarski1955}]\label{prop:tarski}
	Let $(L,\sqsubseteq)$ be a complete lattice and $f \colon L \to L$ be a monotone function.
	Then $f$ has a least fixpoint that is characterised by $\mu f = \bigsqcap \{ x \in L \mid f(x) \sqsubseteq x \}$.
\end{proposition}
Additionally, we need a lemma about the relationship of least fixpoints in different partial orders.
This lemma is a variant of a similar transfer lemma found e.g. in \cite{Apt1986}.
\begin{lemma}\label{lemma:fixpointtranslation}
	Let $(L,\sqsubseteq)$ and $(L',\sqsubseteq')$ be partial orders with functions $f \colon L\to L$, $f' \colon L'\to L'$ and $\mu f$ be the least fixpoint of $f$.
	Let further $h \colon L\to L'$ be a bijective, monotone function with $h\circ f = f'\circ h$.
	Then $\mu f' = h(\mu f)$.
\end{lemma}
\begin{proof}
	We first show that $h(\mu f)$ is a fixpoint of $f'$.
	Since $\mu f$ is a fixpoint of $f$, we have $f'(h(\mu f))=f'\circ h(\mu f)=h\circ f(\mu f)=h(f(\mu f))=h(\mu f)$, i.e. $h(\mu f)$ is a fixpoint of $f'$.
	
	It remains to show that $h(\mu f)$ is the \emph{least} fixpoint of $f'$. Therefore, let $y$ be an arbitrary fixpoint of $f'$. We show that $h(\mu f)\sqsubseteq y$. 
	Since $y$ is a fixpoint of $f'$, we have $f(h^{-1}(y))=h^{-1}\circ h \circ f \circ h^{-1}(y)=h^{-1}\circ f' \circ h \circ h^{-1}(y)=h^{-1}(f'(y))=h^{-1}(y)$, i.e. $h^{-1}(y)$ is a fixpoint of $f$. Since $\mu f$ is the least fixpoint of $f$, we have $\mu f\sqsubseteq h^{-1}(y)$. Since $h$ is monotone, we infer that $h(\mu f)\sqsubseteq' h(h^{-1}(y))=y$.
	Since $y$ was an arbitrary fixpoint of $f'$, the fixpoint $h(\mu f)$ must be the least fixpoint of $f'$.
	\hfill \qed
\end{proof}
In this paper, we consider complete lattices of the form $(2^A,\subseteq)$ for a set $A$.
In these lattices, greatest lower and least upper bounds are given by intersections and unions over sets, respectively.

\section{Proofs from \cref{sec:logic}}\label{app:logicproofs}

\begin{lemma}\label{lem:monotone}
	Let $S = (G,\mathcal{V},X,\varphi)$ where $G=(V, l, (\rightarrow^d)_{d\in\moves}, \curvearrowright)$ is an execution graph, $\mathcal{V}$ is a fixpoint variable assignment, $X$ is a fixpoint variable and $\varphi$ is a well-formed \logicname\ formula in positive normal form.
	Then, the function $\alpha_S: 2^V \to 2^V$ with $\alpha_S(M) = \llbracket \varphi \rrbracket^G_{\mathcal{V}[X \mapsto M]}$ is monotone.
\end{lemma}
\begin{proof}
	The proof is by induction on the structure of $\varphi$.
	\begin{itemize}
		\item \underline{$\varphi = \ap$ for $\ap \in \AP$:} Since $\alpha_S$ is constant in this case, we have $\alpha_S(M) = \alpha_S(M')$ for $M \subseteq M'$.
		\item \underline{$\varphi = \lnot \ap$ for $\ap \in \AP$:} Analogous to the previous case.
		\item \underline{$\varphi = Y$ for a fixpoint variable $Y$:} For $Y \neq X$, $\alpha_S$ is constant in this case, and we have $\alpha_S(M) = \alpha_S(M')$ for $M \subseteq M'$.
		For $Y = X$, we have $\alpha_S(M) = M \subseteq M' = \alpha_S(M')$ for $M \subseteq M'$.
		\item \underline{$\varphi = \varphi_1 \lor \varphi_2$:} Let $S_1 = (G,\mathcal{V},X,\varphi_1)$ and $S_2 = (G,\mathcal{V},X,\varphi_2)$.
		We have $\alpha_S(M) = \alpha_{S_1}(M) \cup \alpha_{S_2}(M) \subseteq \alpha_{S_1}(M') \cup \alpha_{S_2}(M') = \alpha_S(M')$ for $M \subseteq M'$ by the induction hypothesis.
		\item \underline{$\varphi = \varphi_1 \land \varphi_2$:} Analogous to the previous case.
		\item \underline{$\varphi = \bigcirc^f\varphi_1$ for $f \in \{g,\uparrow,a,-,p,c\}$:} Let $S_1 = (G,\mathcal{V},X,\varphi_1)$.
		We have 
		\begin{align*}
			\alpha_S(M) & = \{x \in V \mid \Succ_f^G(x) \textit{ is defined and } \Succ_f^G(x) \in \alpha_{S_1}(M) \}\\
			& \subseteq \{x \in V \mid \Succ_f^G(x) \textit{ is defined and } \Succ_f^G(x) \in \alpha_{S_1}(M') \} \\
			& = \alpha_S(M')
		\end{align*}
		for $M \subseteq M'$ by the induction hypothesis.
		\item \underline{$\varphi = \bigcirc^{\bar{f}}\varphi_1$ for $f \in \{g,\uparrow,a,-,p,c\}$:} Analogous to the previous case.
		\item \underline{$\varphi = \mu Y.\varphi_1$:} 
		For $Y = X$, $\alpha_S$ is constant and we have $\alpha_S(M) = \alpha_S(M')$ for $M \subseteq M'$.
		For $Y \neq X$, let $S_1^{M''} = (G,\mathcal{V}[Y \mapsto M''],X,\varphi_1)$.
		We have 
		\begin{align*}
			\alpha_S(M) & = \bigcap \{ M'' \subseteq V \mid \llbracket \varphi_1 \rrbracket_{\mathcal{V}[X \mapsto M][Y \mapsto M'']}^G \subseteq M'' \}\\
			& = \bigcap \{ M'' \subseteq V \mid \alpha_{S_1^{M''}}(M) \subseteq M'' \}\\
			& \overset{(*)}{\subseteq} \bigcap \{ M'' \subseteq V \mid \alpha_{S_1^{M''}}(M') \subseteq M'' \}\\
			& = \bigcap \{ M'' \subseteq V \mid \llbracket \varphi_1 \rrbracket_{\mathcal{V}[X \mapsto M'][Y \mapsto M'']}^G \subseteq M'' \}\\
			& = \alpha_S(M')	
		\end{align*}
		for $M \subseteq M'$.
		In step $(*)$, the induction hypothesis implies $\alpha_{S_1^{M''}}(M) \subseteq \alpha_{S_1^{M''}}(M')$ for all $M'' \subseteq V$, which then means that $\{ M'' \subseteq V \mid \alpha_{S_1^{M''}}(M) \subseteq M'' \} \supseteq \{ M'' \subseteq V \mid \alpha_{S_1^{M''}}(M') \subseteq M'' \}$ which in turn implies the inclusion $(*)$.
		\item \underline{$\varphi = \nu Y.\varphi_1$:} Analogous to the previous case.
		\hfill \qed
	\end{itemize}
\end{proof}

Using \cref{prop:tarski}, a corollary from this lemma is:

\begin{corollary}\label{cor:fixpoint-characterization}
	Let $S = (G,\mathcal{V},X,\varphi)$ where $G=(V, l, (\rightarrow^d)_{d\in\moves},\curvearrowright)$ is an execution graph, $\mathcal{V}$ is a fixpoint variable assignment, $X$ is a fixpoint variable and $\varphi$ is a well-formed \logicname\ formula in positive normal form.
	Then, $\llbracket \mu X. \varphi \rrbracket_\mathcal{V}^G$ is the least fixpoint of $\alpha_S$.
\end{corollary}

\section{Properties of Successor Functions}\label{app:successorproperties}
We establish some properties of the successor functions defined on execution graphs that are used in some of the proofs in this paper.

\begin{lemma}\label{lemma:successorcharacterization}
	Let $G=(V, l, (\rightarrow^d)_{d\in\moves},\curvearrowright)$ be an execution graph.
	\begin{enumerate}
		\item[(i)] For all $y\in V$, there is $z\in V$ with $z\curvearrowright y$ iff $y$ has a $\ret$-predecessor $x$. In this case we have $z=\Succ^G_\caller(x)$.
		\item[(ii)] For all $x,y\in V$ with $y=\Succ^G_\abstractsucc(x)$, the caller of $x$ is defined iff the caller of $y$ is defined and in this case $\Succ^G_\caller(x)=\Succ^G_\caller(y)$.
		\item[(iii)] For all $x,y\in V$ with $x\rightarrow^\intern y$, $x\rightarrow^\call y$ or $x\curvearrowright y$, the parent of $x$ is defined iff the parent of $y$ is defined and in this case $\Succ^G_\parent(x)=\Succ^G_\parent(y)$.
	\end{enumerate}
\end{lemma}
\begin{proof}
	\begin{enumerate}
		\item[(i)]
		Let $y\in V$ be a node.
		
		For the first direction, assume that there is a node $z\in V$ with $z\curvearrowright y$.
		Then there is a path from $z$ to $y$ following only $\moves\setminus\{\spawn\}$-successors such that the number $n$ of $\call$-moves on that path is equal to the number of $\ret$-moves on that path.
		Since $y\neq z$ and $z$ has a $\call$-successor, we have $n>0$. Since $y$ is defined as the node such that this path has minimal length, the predecessor of $y$ must be a $\ret$-predecessor.
		
		For the other direction, assume that $y$ has a $\ret$-predecessor $x$.
		Since $(v_0,y)\in(\bigcup\limits_{d\in\moves}\rightarrow^d)^*$, there is a node $u\in V$ that is either $v_0$ or has a $\spawn$-predecessor such that there is a path $\pi$ from $u$ to $y$ following only $\moves\setminus\{\spawn\}$-successors.
		Since $y$ is a $\ret$-successor and the number of $\call$-moves on $\pi$ has to be greater or equal to the number of $\ret$-moves on $\pi$, we can consider the last node $z$ on $\pi$ with a $\call$-successor $z'$ such that the number of $\call$-moves on the path $\pi$ between $z$ and $y$ is equal to the number of $\ret$-moves on $\pi$ between $z$ and $y$.
		Then we have $z\curvearrowright y$.
		
		It remains to show that $z=\Succ^G_\caller(x)$.
		Let $n$ be the number of $\call$-transitions on the path $\pi$ from the $\call$-successor $z'$ of $z$ to $x$ following only $\moves\setminus\{\spawn\}$-transitions. 
		Since $z\curvearrowright y$, $z'$ is the $\call$-successor of $z$ and $y$ is the $\ret$-successor of $x$, this is also the number of $\ret$-transitions on $\pi$.
        We now show by induction on $n$ that we can transform the path $\pi$ to a path $\pi'$ from $z'$ to $x$ following abstract successors. 
        Since $z\rightarrow^\call z'$, this implies in particular that $z=\Succ^G_\caller(x)$.
        
        If $n=0$, the path $\pi$ follows only $\intern$-successors, i.e. it is also a path following abstract successors.
        
        If $n>0$, let $c$ and $c'$ be the first nodes on the path $\pi$ with $c\rightarrow^\call c'$.
        Since the number of $\call$-moves on $\pi$ is equal to the number of $\ret$-moves on $\pi$, there is a node $r\in V$ on the path $\pi$ between $c$ and $x$ with $c\curvearrowright r$.
        Thus, $\pi$ can be written as $\pi=\pi_1\pi_2\pi_3$ for paths $\pi_1$ from $z'$ to $c$, $\pi_2$ from $c$ to $r$ and $\pi_3$ from $r$ to $x$ following $\moves\setminus\{\spawn\}$-transitions. 
		By construction, there are no $\call$-moves on $\pi_1$.
		Moreover, there are no $\ret$-moves on $\pi_1$ since this would imply $z\curvearrowright v$ for a node $v\neq y$.	
		Thus, $\pi_1$ is also a path following abstract successors.
		Moreover, the number $m$ of $\call$-moves on $\pi_3$ is also equal to the number of $\ret$-moves on $\pi_3$, since the same holds true for $\pi$, $\pi_1$ and $\pi_2$.
		Since we clearly have $m<n$, we can transform $\pi_3$ to a path $\pi_a$ from $r$ to $x$ following abstract successors by the induction hypothesis. 
	    Since $r$ is the abstract successor of $c$, the concatenation of the paths $\pi_1$ and $\pi_a$ is thus a path from $z'$ to $x$ following abstract successors. 
		
		\item[(ii)]
		$\Succ^G_\caller(x)$ is defined iff there is a node $z\in V$ with a $\call$-successor $z'$ such that there is a path from $z'$ to $x$ following abstract successors.
		Since  $y=\Succ^G_\abstractsucc(x)$ and abstract successors are uniquely determined, we can demand equivalently that there is a node $z\in V$ with a $\call$-successor $z'$ such that there is a path from $z'$ to $y$ following abstract successors, i.e. $\Succ^G_\caller(y)$ is defined and in this case $\Succ^G_\caller(y)=z=\Succ^G_\caller(x)$.
		
		\item[(iii)]
		$\Succ^G_\parent(x)$ is defined iff there is a node $z\in V$ with a $\spawn$-successor $z'$ such that there is a path $\pi$ from $z'$ to $x$ following only $\moves\setminus\{\spawn\}$-transitions.
		If $x\curvearrowright y$, there is also a path $\pi'$ from $x$ to $y$ following only $\moves\setminus\{\spawn\}$-transitions. 
		This also holds trivially if $x\rightarrow^\intern y$ or $x\rightarrow^\call y$.
		Thus, since $\moves\setminus\{\spawn\}$-successors are uniquely determined, we can demand equivalently that there is a node $z\in V$ with a $\spawn$-successor $z'$ such that there is a path $\pi$ from $z'$ to $y$ following only $\moves\setminus\{\spawn\}$-transitions, i.e. $\Succ^G_\parent(y)$ is defined and in this case $\Succ^G_\parent(y)=z=\Succ^G_\parent(x)$.
		\hfill \qed
	\end{enumerate}
\end{proof}

\section{Proofs from \cref{sec:treesemantics}}\label{app:treesemanticsproofs}

\begin{proof}[\textbf{Proof of \cref{lemma:successorequivalence}}]
Let $x\in V$ be an arbitrary node.
Since $\delta_G(y)$ is defined for all nodes $y\in V$, we only have to show that $\Succ_f^G(x)$ is defined iff $\Succ_f^{\mathcal{T}}(\delta_G(x))$ is defined and that $\delta_G(\Succ_f^G(x))=\Succ_f^{\mathcal{T}}(\delta_G(x))$ holds in this case.
In order to improve readability, we also write $\Succ^G_f(x)=\Succ^G_f(y)$ for $f\in\{\globalsucc,\globalpred,\abstractsucc,\caller,\parent,\child\}$ and $y\in V$, if both $\Succ^G_f(x)$ and $\Succ^G_f(y)$ are undefined.
We show the claim for each successor type separately.
\begin{itemize}
	\item 
	$\underline{f=\abstractsucc:}$ 
	The claim is shown by a case distinction on $d(x)$.
	
	If $d(x)\in\{\intern,\spawn\}$, then $x$ has an $\intern$-successor $y$, i.e. $\Succ^G_\abstractsucc(x)=y$. 
	Moreover, we have $\delta_G(y)=\delta_G(x)\cdot 0$ and thus $\Succ_\abstractsucc^{\mathcal{T}}(\delta_G(x))=\delta_G(x)\cdot 0  =\delta_G(y)=\delta_G(\Succ^G_\abstractsucc(x))$.
	
	If $d(x)=\callret$, there is $y\in V$ with $x\curvearrowright y$, i.e. $\Succ^G_\abstractsucc(x)=y$.
	Moreover, we have $\delta_G(y)=\delta_G(x)\cdot 1$ and thus
	 $\Succ_\abstractsucc^{\mathcal{T}}(\delta_G(x))=\delta_G(x)\cdot 1=\delta_G(y)=\delta_G(\Succ^G_\abstractsucc(x))$.
	
	If $d(x)\in\{\call,\ret,\labelend\}$, there is no $y\in V$ with $x\curvearrowright y$ and $x$ has no $\intern$-successor, i.e. $\Succ^G_\abstractsucc(x)$ is undefined. Moreover, $\Succ_\abstractsucc^{\mathcal{T}}(\delta_G(x))$ is also undefined in this case.
	
	Thus, we have established $\delta_G\circ\Succ^G_\abstractsucc(x)=\Succ_\abstractsucc^{\mathcal{T}}\circ \delta_G(x)$ in each of the cases.
	
	\item
	$\underline{f=\caller:}$ 
	Since $G$ is an execution graph, we have $(v_0,x)\in (\bigcup\limits_{d\in\moves}\rightarrow^d)^*$. 
	Thus, let $\pi$ be a path from $v_0$ to $x$ following $\moves$-successors. 
	We show the claim by induction on the length $n$ of $\pi$.
	
	If $n=0$, we have $x=v_0$. In this case $x$ has no predecessor and $\Succ^G_\caller(x)$ is undefined.
	Moreover, we have $p(x)=\bot$ and thus $\Succ_\caller^{\mathcal{T}}(\delta_G(x))$ is also undefined.	
	
	If $n>0$, let $y$ be the predecessor node of $x$ in $\pi$.
	We show the claim by a case distinction based on what type of predecessor $y$ is.
	
	If $y\rightarrow^\call x$, we have $\Succ^G_\caller(x)=y$ and $\delta_G(x)=\delta_G(y)\cdot 0$, i.e. $\delta_G(y)$ is the parent node of $\delta_G(x)$.
	Moreover, we have $p(x)=\call$ and hence $\Succ_\caller^{\mathcal{T}}(\delta_G(x))=\delta_G(y)=\delta_G(\Succ^G_\caller(x))$.
	
	If $y\rightarrow^\intern x$, then $x=\Succ^G_\abstractsucc(y)$. 
	Thus, by \cref{lemma:successorcharacterization}(ii), $\Succ^G_\caller(x)=\Succ^G_\caller(y)$.
	Moreover, we have $\delta_G(x)=\delta_G(y)\cdot 0$, i.e. $\delta_G(y)$ is the parent node of $\delta_G(x)$ and we have $p(x)=\intern$, i.e. $\Succ_\caller^{\mathcal{T}}(\delta_G(x))=\Succ_\caller^{\mathcal{T}}(\delta_G(y))$. 
	By the induction hypothesis, we obtain $\delta_G(\Succ^G_\caller(x))=\delta_G(\Succ^G_\caller(y))=\Succ_\caller^{\mathcal{T}}(\delta_G(y))=\Succ_\caller^{\mathcal{T}}(\delta_G(x))$.
		
	If $y\rightarrow^\ret x$, then by \cref{lemma:successorcharacterization}(i), there is $z\in V$ with $z\curvearrowright x$. 
	We then have $x=\Succ^G_\abstractsucc(z)$ and $\delta_G(x)=\delta_G(z)\cdot 1$, i.e. $\delta_G(z)$ is the parent node of $\delta_G(x)$. 
	The claim follows as in the previous case.
		
	If $y\rightarrow^\spawn x$, then $\Succ^G_\caller(x)$ is undefined since $x$ has no $\call-$, $\intern$- or $\ret$-predecessor, i.e. there is no $y\in V$ such that $x=\Succ^G_\abstractsucc(y)$ or $x$ is a $\call$-successor of $y$.
	Moreover,
	we have $p(x)=\spawn$ and thus $\Succ_\caller^{\mathcal{T}}(\delta_G(x))$ is also undefined.	
	
	Thus, we have established $\delta_G\circ\Succ^G_\caller(x)=\Succ_\caller^{\mathcal{T}}\circ \delta_G(x)$ in each of the cases.
	
	\item 
	$\underline{f=\globalsucc:}$ 
	The claim is shown by a case distinction on $d(x)$.
	
	If $d(x)\in\{\intern,\call,\callret,\spawn\}$, then $x$ has an $\intern$- or $\call$-successor $y$, i.e. $\Succ^G_\globalsucc(x)=y$. 
	Moreover, we have $\delta_G(y)=\delta_G(x)\cdot 0$ and thus $\Succ_\globalsucc^{\mathcal{T}}(\delta_G(x))=\delta_G(x)\cdot 0=\delta_G(y)=\delta_G(\Succ^G_\globalsucc(x))$.
	
	If $d(x)=\ret$, then $x$ has a $\ret$-successor $y$, i.e. $\Succ^G_\globalsucc(x)=y$. 
	By \cref{lemma:successorcharacterization}(i), there is a node $z\in V$ with $z\curvearrowright y$ and $z=\Succ^G_\caller(x)$ and thus $y=\Succ^G_\abstractsucc(z)=\Succ^G_\abstractsucc(\Succ^G_\caller(x))$. 
	Since we have already shown the claim for the abstract successor and the caller, we conclude that $\Succ_\globalsucc^{\mathcal{T}}(\delta_G(x))=\Succ_\abstractsucc^{\mathcal{T}}(\Succ_\caller^{\mathcal{T}}(\delta_G(x)))=\Succ_\abstractsucc^{\mathcal{T}}(\delta_G(\Succ^G_\caller(x)))=\delta_G(\Succ^G_\abstractsucc(\Succ^G_\caller(x)))=\\
	\delta_G(y)=\delta_G(\Succ^G_\globalsucc(x))$.
	
	If $d(x)=\labelend$, then $x$ has no successor, i.e. $\Succ^G_\globalsucc(x)$ is undefined. 
	Moreover, \\$\Succ_\globalsucc^{\mathcal{T}}(\delta_G(x))$ is also undefined in this case.
	
	Thus, we have established $\delta_G\circ\Succ^G_\globalsucc(x)=\Succ_\globalsucc^{\mathcal{T}}\circ\delta_G(x)$ in each of the cases.
	
	\item
	$\underline{f=\ \globalpred:}$ 
	We show the claim by a case distinction on $p(x)$.
	
	If $p(x)\in\{\intern,\call\}$, then $x$ has a $p(x)$-predecessor $y$, i.e. $\Succ^G_\globalpred(x)=y$. 
	Moreover, $\delta_G(y)$ is the parent node of $\delta_G(x)=\delta_G(y)\cdot 0$. 
	Thus, we have $\delta_G(\Succ^G_\globalpred(x))=\delta_G(y)=\Succ_\globalpred^{\mathcal{T}}(\delta_G(x))$.
	
	If $p(x)=\ret$, then $x$ has a $\ret$-predecessor $y$, i.e. $\Succ^G_\globalpred(x)=y$, and by \cref{lemma:successorcharacterization}(i) there is a node $z\in V$ with $z\curvearrowright x$ and $\Succ^G_\caller(y)=z$.
	Using the claim for $f=\caller$, which we have already seen, we thus have $\delta_G(z)=\delta_G(\Succ^G_\caller(y))=\Succ^{\mathcal{T}}_\caller(\delta_G(y))$.
	By definition, the caller predecessor of $\delta_G(y)$ has to be a node with a $\call$-child $\delta_G(z')$ such that $\delta_G(z')$ is an $\{\intern,\ret\}$-ancestor of $\delta_G(y)$.
	Moreover, since $x$ is a $\ret$-successor of $y$, we have $d(y)=\ret$ and hence $\delta_G(y)$ is a leaf.
	Thus, $\delta_G(y)$ is the $\{\intern,\ret\}$-descendant leaf of $\delta_G(z')$, which is the left child of the parent node $\delta_G(z)$ of $\delta_G(x)=\delta_G(z)\cdot 1$, i.e. $\Succ^{\mathcal{T}}_\globalpred(\delta_G(x))=\delta_G(y)$.
	Hence,
	$\delta_G(\Succ^G_\globalpred(x))=\delta_G(y)=\Succ_\globalpred^{\mathcal{T}}(\delta_G(x))$.
	
	If $p(x)\in\{\spawn,\bot\}$, then $x$ has no $\intern$-, $\call$- or $\ret$-predecessor, i.e. $\Succ^G_\uparrow(x)$ is undefined.
	Moreover, $\Succ_\globalpred^{\mathcal{T}}(\delta_G(x))$ is also undefined in this case.
	
	Thus, we have established $\delta_G\circ\Succ^G_\globalpred(x)=\Succ_\globalpred^{\mathcal{T}}\circ \delta_G(x)$ in each of the cases.	
	
	\item
	$\underline{f=\parent:}$ 
	Since $G$ is an execution graph, we have $(v_0,x)\in (\bigcup\limits_{d\in\moves}\rightarrow^d)^*$. 
	Thus, let $\pi$ be a path from $v_0$ to $x$ following $\moves$-successors. 
	We show the claim by induction via the length $n$ of $\pi$.
	
	If $n=0$, we have $x=v_0$. 
	In this case $x$ has no predecessor and $\Succ^G_\parent(x)$ is undefined.
	Moreover, we have $p(x)=\bot$ and thus $\Succ_\parent^{\mathcal{T}}(\delta_G(x))$ is also undefined.	
	
	If $n>0$, let $y$ be the predecessor node of $x$ in $\pi$. 
	We show the claim by a case distinction based on what type of predecessor $y$ is.
	
	If $y\rightarrow^\ret x$, by \cref{lemma:successorcharacterization}(i), there is $z\in V$ with $z\curvearrowright x$.
	Hence, if $y\rightarrow^\intern x$, $y\rightarrow^\call x$ or $y\rightarrow^\ret x$, there is a node $z\in V$ with $z\rightarrow^\intern x$, $z\rightarrow^\call x$ or $z\curvearrowright x$.
	Thus, by \cref{lemma:successorcharacterization}(iii), $\Succ^G_\parent(x)=\Succ^G_\parent(z)$.
	Moreover, $\delta_G(z)$ is the parent node of $\delta_G(x)$ and we have
	$p(x)\in\{\intern,\call,\ret\}$, i.e. $\Succ_\parent^{\mathcal{T}}(\delta_G(x))=\Succ_\parent^{\mathcal{T}}(\delta_G(z))$. 
	By the induction hypothesis, we obtain\\ $\delta_G(\Succ^G_\parent(x))=\delta_G(\Succ^G_\parent(z))=\Succ_\parent^{\mathcal{T}}(\delta_G(z))=\Succ_\parent^{\mathcal{T}}(\delta_G(x))$.
	
	If $y\rightarrow^\spawn x$, we have $\Succ^G_\parent(x)=y$ and $\delta_G(x)=\delta_G(y)\cdot 1$, i.e. $\delta_G(y)$ is the parent node of $\delta_G(x)$.
	Moreover, we have $p(x)=\spawn$ and hence $\Succ_\parent^{\mathcal{T}}(\delta_G(x))=\delta_G(y)=\delta_G(\Succ^G_\parent(x))$.
	
	Thus, we have established $\delta_G\circ\Succ^G_\parent(x)=\Succ_\parent^{\mathcal{T}}\circ \delta_G(x)$ in each of the cases.
	
	\item
	$\underline{f=\child:}$ 
	We distinguish two cases for $d(x)$.
	
	If $d(x)=\spawn$, then $x$ has a $\spawn$-successor $y$, i.e. we have $\Succ^G_\child(x)=y$.
	Moreover, we have $\delta_G(y)=\delta_G(x)\cdot 1$ and thus $\Succ_\child^{\mathcal{T}}(\delta_G(x))=\delta_G(y)=\delta_G(\Succ^G_\child(x))$.
	
	If $d(x)\neq \spawn$, then $x$ has no $\spawn$-successor, i.e. $\Succ^G_\child(x)$ is undefined.
	Moreover, $\delta_G(x)$ is also undefined in this case.
	
	Thus, we have established $\delta_G\circ\Succ^G_\child(x)=\Succ_\child^{\mathcal{T}}\circ \delta_G(x)$ in both cases.
	\hfill \qed
\end{itemize}
\end{proof}

\section{Proofs from \cref{sec:decisionproblems}}\label{app:decisionproblemsproofs}

\begin{proof}[\textbf{Proof of \cref{thm:extreeautomaton}}]
	We first show that $\mathcal{A}_\ExTrees$ accepts all execution trees. 

	For this, let $G=(V, l, (\rightarrow^d)_{d\in\moves},\curvearrowright)$ be an execution graph and $\mathcal{T}(G)=(T,r)$ be the tree representation of $G$.
	We inductively define a map $r_R\colon V\to \{0,1\}$ as follows.
	First, we set $r_R(v_0):=0$.
	Then, for each node $x\in V$,
	\begin{itemize}
		\item if there is a node $y\in V$ such that $y$ is an $\intern$-predecessor of $x$ or $y\curvearrowright x$ (the latter holds by \cref{lemma:successorcharacterization}(i) iff $x$ has a $\ret$-predecessor), we set $r_R(x):=r_R(y)$,
		\item if $x$ (i) has a $\spawn$-predecessor or (ii) it has a $\call$-predecessor $y$ and there is no node $z\in V$ with $y\curvearrowright z$, we set $r_R(x):=0$ and
		\item if $x$ has a $\call$-predecessor $y$ and there is a node $z\in V$ with $y\curvearrowright z$, we set $r_R(x):=1$.
	\end{itemize} 
	Next, we define a map $r_A \colon T\to Q$ by $r_A(\delta_G(x)):=(p(x),r_R(x))$ and show the following claim:
	
	\underline{Claim:} $r_A$ is an accepting $(\varepsilon, q_0)$-run of $\mathcal{A}_\ExTrees$ over $\mathcal{T}(G)$.
	
	We first show that $r_A $ is an $(\varepsilon, q_0)$-run of $\mathcal{A}_\ExTrees$ over $\mathcal{T}(G)$.
	For the initial node, we have $p(v_0)=\bot$ and $r_R(v_0)=0$ and thus $r_A(\varepsilon)=r_A(\delta_G(v_0))=(\bot,0)=q_0$.
		
	Let $t=\delta_G(x)\in T$ be an arbitrary node with $r(t)=(l(x),d(x),p(x))$ and $r_A(t)=(p(x),c)$.
	By a case distinction on $d(x)$, we show that the children of $t$ satisfy the transition function in this node.
	\begin{itemize}
		\item 
		If $d(x)=\intern$, then $x$ has exactly one $\intern$-successor $y$ with $p(y)=\intern$, $r_R(y)=r_R(x)=c$ and $\delta_G(y)=t\cdot 0$.
		Thus, $\{(0,r_A(t\cdot 0))\}=\{(0, (\intern,c))\}$ satisfies $\rho((p(x),c),r(t))$.
		 
		\item 
		If $d(x)=\call$, then $x$ has exactly one $\call$-successor $y$ with $p(y)=\call$, there is no node $z\in V$ with $x\curvearrowright z$, i.e. $r_R(x)=0$, and $\delta_G(y)=t\cdot 0$.
		Thus $\{(0,r_A(t\cdot 0))\}=\{(0, (\call,0))\}$ satisfies $\rho((p(x),c),r(t))$.
		
		\item 
		If $d(x)=\callret$, then $x$ has a $\call$-successor $y$ with $p(y)=\call$, there is a node $z\in V$ with $x\curvearrowright z$ and by \cref{lemma:successorcharacterization}(i), there is a node $z'\in V$ with $z'\rightarrow^\ret z$, i.e. $p(z)=\ret$.	
		Thus, we have $r_R(y)=1$, $\delta_G(y)=t\cdot 0$, $r_R(z)=r_R(x)=c$, and $\delta_G(z)=t\cdot 1$.
		Thus, $\{(0,r_A(t\cdot 0)), (1, r_A(t\cdot 1))\}=\{(0, (\call,1)), (1, (\ret,c))\}$ satisfies $\rho((p(x),c),r(t))$.
		
		\item 
		If $d(x)=\spawn$, then $x$ has an $\intern$-successor $y$ and a $\spawn$-successor $z$ with $p(y)=\intern$, $r_R(y)=r_R(x)=c$, $\delta_G(y)=t\cdot 0$, $p(z)=\spawn$, $r_R(z)=0$, and $\delta_G(z)=t\cdot 1$.
		Thus, $\{(0,r_A(t\cdot 0)), (1, r_A(t\cdot 1))\}=\{(0, (\intern,c)), (1, (\spawn, 0))\}$ satisfies $\rho((p(x),c),r(t))$.
		
		\item 
		If $d(x)=\ret$, then $x$ has a $\ret$-successor $y$.
		By \cref{lemma:successorcharacterization}(i), there is a node $z\in V$ with $z\curvearrowright y$ and $\Succ_\caller^G(x)=z$, i.e. there is a path from the $\call$-successor $z'$ of $z$ to $x$ following abstract successors.
		Then we clearly have $c=r_R(x)=r_R(z')=1$ by construction.
		Thus, $\emptyset$ satisfies $\true=\rho((p(x),c),r(t))$.	
		 
		\item 
		If $d(x)=\labelend$, then $x$ has no successors. 
		Assume towards contradiction that $r_R(x)=1$.
		Clearly, by construction, there is a node $z\in V$ with $z\curvearrowright y$ and a path from the $\call$-successor $z'$ of $z$ to $x$ following abstract successors.
		Thus, $\delta_G(x)$ is the $\{\intern,\ret\}$-descendant leaf of the left child $\delta_G(z')=\delta_G(z)\cdot 0$ of the parent node $\delta_G(z)$ of $\delta_G(y)=\delta_G(z)\cdot 1$.
		Then we have $\Succ_\globalpred^{\mathcal{T}(G)}(\delta_G(y))=\delta_G(x)$.
		Using \cref{lemma:successorequivalence}, we obtain \\
		$\delta_G(\Succ_\globalpred^G(y))=\Succ_\globalpred^{\mathcal{T}(G)}(\delta_G(y))=\delta_G(x)$ and thus $\Succ^G_\globalpred(y)=x$ since $\delta_G$ is injective.
		This means that $y$ is a successor of $x$, which contradicts our assumption that $x$ has no successor.
		Thus, $c=r_R(x)=0$, and $\emptyset$ satisfies $\true=\rho((p(x),c),r(t))$.	
	\end{itemize}
	Thus, $r_A$ is an $(\varepsilon, q_0)$-run of $\mathcal{A}_\ExTrees$ over $\mathcal{T}(G)$.
	 
	It remains to show that the run is accepting.
 	Assume towards contradiction that there is an infinite path $\delta_G(x_0)\delta_G(x_1)\dots$ in $T$ where the priority $0$ occurs only finitely often, i.e. there is a minimal $i>0$ such that for all $j\geq i$ we have $r_R(x_j)=1$.
 	By construction, there must be a node $y\in V$ with $x_{i-1}\curvearrowright y$ such that $x_{i-1}$ is the $\call$-predecessor of $x_i$ and we have $r_R(y)=r_R(x_{i-1})=0$.
 	Thus, there is a finite path in $G$ from $x_{i-1}$ to $y\neq x_{i-1}$ following only $\moves\setminus\{\spawn\}$-successors such that the number of $\call$-moves on the path is equal to the number of $\ret$-moves on the path.
	On the other hand, the infinite path in the tree $T$ up from $\delta_G(x_{i})$ cannot contain a $\spawn$-child by construction and thus provides an infinite path in $G$ starting in $x_i$ and following only $\call$- or abstract successors and thus also an infinite path in $G$ starting in $x_{i-1}$ following only $\moves\setminus\{\spawn\}$-successors.
	Since $\moves\setminus\{\spawn\}$-successors are uniquely determined, $y$ must be contained in this path.
	But since the paths between nodes $z$ and $z'$ with $z\curvearrowright z'$ are the minimal paths of length greater than zero from $z$ such that the number of $\call$-moves is equal to the number of $\ret$-moves, the number of $\call$-moves on the infinite path up from $x_i$ so far is always greater or equal to the number of $\ret$-moves on this path so far.
	Since $x_i$ is the $\call$-successor of $x_{i-1}$, the number of $\call$-moves on the infinite path in $G$ from $x_{i-1}$ so far is always greater than the number of $\ret$-moves on this path so far after the first move.
	This contradicts the fact that $y\neq x_{i-1}$ is contained in this path.
	Thus, for all infinite paths in $T$, the priority $0$ occurs infinitely often, and the run is accepting.
	 	 
	We now show that all trees accepted by $\mathcal{A}_\ExTrees$ are execution trees.
	For this, let $\mathcal{T}_A=(T,r)$ be a tree accepted by $\mathcal{A_\ExTrees}$, witnessed by the accepting $(\varepsilon, q_0)$-run $r_A$ of $\mathcal{A}_\ExTrees$ over $\mathcal{T}_A$.
	We define an execution graph $G=(V, l, (\rightarrow^d)_{d\in\moves},\curvearrowright)$ and show that $\mathcal{T}_A$ is the tree representation of $G$.

	The components of $G$ are defined as follows. 
	First, we set $V:=T$ and $l(t) := l_t$ where $r(t) = (l_t,d_t,p_t)$.
	For the definition of the transitions of $G$, let $t\in V$ be a node with $r_A(t)=q$ and $r(t)=(l',d,p)$.
	Outgoing transitions in $t$ are defined based on $d$.
	\begin{itemize}
		\item If $d\in\{\intern,\call\}$, then $\ar(r(t))=1$, i.e. $t$ has a child $t\cdot 0$ and we include $t\rightarrow^d t\cdot 0$.
		
		\item If $d\in\{\callret,\spawn\}$, then $\ar(r(t))=2$, i.e. $t$ has two children $t\cdot 0$ and $t\cdot 1$.
		For $d=\callret$, we include $t\rightarrow^\call t\cdot 0$ and $t\curvearrowright t\cdot 1$.
		For $d=\spawn$, we include $t\rightarrow^\intern t\cdot 0$ and $t\rightarrow^\spawn t\cdot 1$.
	\end{itemize}
	
	In order to define the transition relation $\rightarrow^\ret$, we show by induction on the length of $t$ that for all $t\in V$ with $r_A(t)=(p,c)$:
	
	$(*)$ $c=1$ iff there are $t_1,t_2,t_3\in V$ with $t_1\rightarrow^\call t_2$, $t_1\curvearrowright t_3$ and there is a path from $t_2$ to $t$ following only $\intern$- and $\ret$-children.
	
	In the base case, where $t=\varepsilon$, we have $c=0$ and $t$ has no parent node and no $\call$-predecessor.
	
	In the inductive step, let $t$ be of the form $t'\cdot i$ for $i\in\{0,1\}$ with $r_A(t')=(p',c')$ and $r(t')=(l',d,p'')$.
	The claim is shown by a case distinction on $d$.
	\begin{itemize}
		\item If $d=\intern$, we have $i=0$ and $\rho(r_A(t'),r(t'))=(0,(\intern,c'))$, i.e. $c=c'$.
		Since $t=t'\cdot 0$ is the $\intern$-child of $t'$ and $\intern$- or $\ret$-children are uniquely determined, the required nodes and the path exist for $t$ iff they exist for $t'$, and the latter holds by induction hypothesis iff $c=c'=1$.
		
		\item If $d=\call$, we have $i=0$ and $\rho(r_A(t'),r(t'))=(0,(\call,0))$, i.e. $c=0$.
		Since $t$ is no $\intern$- or $\ret$-child and there is no node $\tilde{t}\in V$ with $t'\curvearrowright \tilde{t}$, the required nodes and the path do not exist for $t'\cdot 0=t$.
		
		\item If $d=\callret$, we have $\rho(r_A(t'),r(t'))=(0,(\call,1))\land (1,(\ret,c'))$.
		If $i=0$, then $c=1$.
		Since  $t'\rightarrow^\call t'\cdot 0$ and $t'\curvearrowright t'\cdot 1$, the nodes $t_1=t'$, $t_2=t'\cdot 0$ and $t_3=t'\cdot 1$ and the empty path from $t_2=t'\cdot 0=t$ to $t$ witness that $(*)$ holds.
		If $i=1$, then $c=c'$ and $t'\curvearrowright t$.
		Since $\intern$- or $\ret$-children are uniquely determined, the required nodes and the path exist for $t$ iff they exist for $t'$, and the latter holds by induction hypothesis iff $c=1$.
		
		\item If $d=\spawn$, we have $\rho(r_A(t'),r(t'))=(0,(\intern,c'))\land (1,(\spawn,0))$.
		If $i=0$, then $c=c'$.
		Since $t'\rightarrow^\intern t'\cdot 0=t$ and $\intern$- or $\ret$-children are uniquely determined, the required nodes and the path exist for $t$ iff they exist for $t'$, and the latter holds by induction hypothesis iff $c=c'=1$.
		If $i=1$, then $c=0$.
		Since $t$ is no $\intern$- or $\ret$-child and it has no $\call$-predecessor, the required nodes and the path do not exist for $t'\cdot 1=t$.
	\end{itemize}
	Given $(*)$, for each node $t\in V$ with $r_A(t)=(p,1)$ and $r(t)=(l',\ret,p')$, there are nodes $t_1,t_2,t_3\in V$ with $t_1\rightarrow^\call t_2$, $t_1\curvearrowright t_3$ and there is a path from $t_2$ to $t$ following only $\intern$- and $\ret$-children.
	We then include the transition $t\rightarrow^\ret t_3$.
	
	We now show that $G$ is indeed an execution graph. 
	For this, we separately check each of the conditions from the definition of execution graphs.
	\begin{enumerate}
	\item
	Clearly, the node $\varepsilon$ has no predecessor with respect to  $(\rightarrow^d)_{d\in\moves}$.
	Moreover, every node $t\neq\varepsilon$ has exactly one predecessor with respect to $(\rightarrow^d)_{d\in\moves\setminus\{\ret\}}$ and $\curvearrowright$.
	Additionally, $t$ can only have a $\ret$-predecessor, if $t$ has a predecessor with respect to $\curvearrowright$.
	Now let $t'\in V$ be a node with $t'\curvearrowright t$.
	We show that $t$ has a unique $\ret$-predecessor in this case.
	We know that $t'$ has a $\call$-successor $t'\cdot 0\in V$.
	Now consider the unique maximal path in $T$ from $t'\cdot 0$ following only $\intern$- and $\ret$-chlidren.
	By $(*)$, we have $r_A(x)=(p,1)$ for some $p\in \moves\cup\{\bot\}$ for all nodes $x$ on the path.
	Since $\Omega(q)=1$ for all states $q$ visited on the path and the run $r_A$ is accepting, the path cannot be infinite.
	Thus, there is a node $x$ in the path that has no $\intern$- or $\ret$-child.
	Assume towards contradiction that $x$ has a $\call$-child $y$.
	For $r_A(x)=(p,c)$ and $r(x)=(l',d,p')$, we must have $d=\call$ in this case.
	Since $y=x\cdot 0$ has to satisfy $\rho((p,c),(l',\call,p'))$, we thus have $c=0$ which contradicts our assumption that $x$ is on the given path.
	Therefore, $x$ also has no $\call$-child, i.e. we have $d\in\{\ret,\labelend\}$.
	Since $\rho(r_A(x),r(x))=\rho((p,1),(l',d,p))$ must be satisfied by the children of $x$, we then have $d=\ret$, since otherwise we had $\rho(r_A(x),r(x))=\false$.
	Thus, $x\rightarrow^\ret t$ holds by construction.
	Clearly, since $\intern$- or $\ret$-children are uniquely determined, we can only have $y\rightarrow^\ret t$ for a node $y\in T$, if $y$ is on the given path and it is a leaf, i.e. if $y=x$.

	\item
	By construction of $V$, there is a path $\pi$ from $\varepsilon$ to $x$ following only $\moves\setminus\{\ret\}$-transitions or nesting edges for all nodes $x\in V$. 
	Consider the first nodes $y,z$ on the path $\pi$ with $y\curvearrowright z$, if they exist.
	As shown in (i), there is a path $\pi'$ from the $\call$-successor $y'$ of $y$ to the $\ret$-predecessor $z'$ of $z$ following only $\intern$- and $\ret$-children, i.e. $\intern$-transitions and nesting edges.
	Now include the concatenation of the paths $y\rightarrow^\call y'$, $\pi'$ and $z'\rightarrow^\ret z$ in the initial path $\pi$ between the nodes $y$ and $z$ and repeat this construction.
	If this provided an infinite procedure, we would obtain an infinite path from $y'$ following only $\call$- or $\intern$-successors.
	However, since $y\curvearrowright z$, all states visited on the path are of the form $(p,1)$, which contradicts the fact that the run $r_A$ is accepting.
	Thus, the given procedure terminates and we finally obtain a path from $\varepsilon$ to $x$ following only $\moves$-transitions.
	
	\item
	By construction, each node clearly  either has 
	(a) exactly one $\intern$-successor and at most one $\spawn$-successor, 
	(b) one $\call$-successor,
	(c) one $\ret$-successor or
	(d) no successors.	
	
	\item Let $x\in V$ be a node.
	In the construction of the path from $\varepsilon$ to $x$ following only $\moves$-transitions in (ii), we never add a $\spawn$-transition on the path between nodes $y$ and $z$ with $y\curvearrowright z$.
	Moreover, for each new nesting edge we add one $\call$-transition and one $\ret$-transition afterwards on the path.
	Thus, on the path from $\varepsilon$ or a node with a $\spawn$-predecessor to $x$, the number of $\call$-moves is greater or equal to the number of $\ret$-moves, since we start without any $\ret$-moves in this construction.
	
	\item
	Let $x\in V$ be a node with a $\call$-successor.
	For the first direction, let $y\in V$ be a node with $x\curvearrowright y$.
	By the construction given in (ii), we obtain a path from $x$ to $y$ following only $\moves\setminus\{\spawn\}$-transitions from $x$ to $y$ such that the number of $\call$-moves on the path is equal to the number of $\ret$-moves on the path.
	Moreover, for all nodes $z$ between $x$ and $y$, the number of $\call$-moves on the path between $x$ and $z$ is greater than the number of $\ret$-moves on this path, i.e. $y$ is given as the node with the stated property such that the witnessing path has minimal length.
	
	For the other direction, let $y\in V$ with $y\neq x$ be a node such that there is a path from $x$ to $y$ following only $\moves\setminus\{\spawn\}$-transitions where the number of $\call$-moves on the path is equal to the number of $\ret$-moves on the path and the path has minimal length.
	Consider a subpath of this path from a node $z$ to a node $z'$ starting with a $\call$-transition, then following only $\intern$-transitions and finally ending with a $\ret$-transition.
	By the construction in (ii), we only obtain such a path if $z\curvearrowright z'$.
	Now remove the nodes between $z$ and $z'$ from the path and repeat this procedure.
	Since the initial path has minimal length, it must end with a $\ret$-transition, it starts with a $\call$-transition and we finally see that $x\curvearrowright y$.
	\end{enumerate}
	Thus, we have shown that $G$ is indeed an execution graph.
	It remains to show that $\mathcal{T}_A$ is the tree representation of $G$.
	
	For this, we show by induction over the length of $t$ that for all $t\in T$ with $r_A(t)=(p,c)$ we have $\delta_G(t)=t$ and either $p\neq \bot$ and $t$ has a $p$-predecessor or $p=\bot$ and $t=\varepsilon$.
	
	In the base case, where $t=\varepsilon$, we have $\delta_G(t)=\varepsilon=t$ since $\varepsilon\in V$ has no predecessor.
	Moreover, we have $r_A(t)=q_0=(\bot,0)$, since $r_A$ is an $(\varepsilon,q_0)$-run, i.e. $p=\bot$.
	
	In the inductive step, let $t$ be of the form $t'\cdot i$ for $i\in\{0,1\}$ with $r(t')=(l',d,p')$.
	Since $T$ is prefix-closed, we have $t'\in T$ and thus, by the induction hypothesis, $\delta_G(t')=t'$.
	Here, we make a case distinction on $d$.
	\begin{itemize}
		\item 	
		If $d\in\{\intern,\call\}$, then $\ar(r(t'))=1$, i.e. $i=0$ since $t=t'\cdot i\in T$.
		Moreover, we have $t'\rightarrow^\intern t$ and thus $\delta_G(t)=\delta_G(t')\cdot 0=t'\cdot 0=t$ and $r_A(t)=(d,c)$ for a $c\in \{0,1\}$ and $t$ has a $d$-predecessor $t'$.
		
		\item 
		If $d=\callret$, then $t'$ has a $\call$-successor $t'\cdot 0$ and we have $t'\curvearrowright t'\cdot 1$.
		If additionally $i=0$, we have $\delta_G(t)=\delta_G(t')\cdot 0=t'\cdot 0=t$ and $r_A(t)=(\call,1)$ and $t$ has a $\call$-predecessor $t'$.
		If instead $i=1$, we have $\delta_G(t)=\delta_G(t')\cdot 1=t'\cdot 1=t$ and $r_A(t)=(\ret,c)$ for a $c\in\{0,1\}$ and since $t'\curvearrowright t'\cdot 1$ and $G$ is an execution graph, $t=t'\cdot 1$ has a $\ret$-predecessor by \cref{lemma:successorcharacterization}(i).
		
		\item 
		If $d=\spawn$, then $t'$ has an $\intern$-successor $t'\cdot 0$ and a $\spawn$-successor $t'\cdot 1$.
		If additionally $i=0$, we have $\delta_G(t)=\delta_G(t')\cdot 0=t'\cdot 0=t$ and $r_A(t)=(\intern,c)$ for a $c\in\{0,1\}$ and $t$ has an $\intern$-predecessor $t'$.
		If instead $i=1$, we have $\delta_G(t)=\delta_G(t')\cdot 1=t'\cdot 1=t$ and $r_A(t)=(\spawn,0)$ and $t$ has a $\spawn$-predecessor $t'$.
		
		\item We cannot have $d\in\{\ret,\labelend\}$, since that would mean $\ar(r(t'))=0$, i.e. $t'$ has no children.
	\end{itemize}
	Moreover, for all nodes $t\in V$ with $r_A(t)=(p,c)$ and $r(t)=(l',d,p')$, the boolean formula $\rho((p,c),(l',d,p'))$ must be satisfied by the children of $t$, i.e we have $p'=p$ since otherwise we had $\rho((p,c),(l',d,p'))=\false$.
	Thus, $p'$ characterizes the predecessor type of $t$ as required. 
	Finally, it is straightforward to see that for all $t\in T$ with $r(t)=(l,d,p)$, $d$ specifies the successor types of $t=\delta_G^{-1}(t)$ as required for tree representations.
	
	Overall, we conclude $\mathcal{T}_A=(T,r)=(\mathit{im}(\delta_G), r)=\mathcal{T}(G)$.
	\hfill \qed
\end{proof}

\begin{proof}[\textbf{Proof of \cref{thm:formulatoautomaton}}]

We prove the theorem by induction on the structure of $\varphi$.
Therefore, we also have do deal with non-closed subformulae and consider valuations to decide whether a subformula is satisfied.
In order to do this in a formal way, we consider automata with special states $X_1,\dots,X_n$, called \emph{holes}, that can be filled with sets of nodes $ L_1,\dots,L_n$ of a given tree.
Intuitively, such an automaton can operate on a tree as before, but when a hole $X_i$ is encountered during a run and we are at the tree node $t$, then we do not continue on the current path and say that it is accepting iff $t\in L_i$.

Formally, let $\mathcal{A} = (Q,q_0,\rho,\Omega)$ be a 2-way alternating tree automaton over $(\Sigma,\ar)$-labelled binary trees with states $q,X_1,...,X_n\in Q$, $\mathcal{T}=(T,l)$ be a $(\Sigma,\ar)$-labelled binary tree, $t\in T$ be a tree node and $L_1,...,L_n\subseteq T$ be sets of tree nodes.
A $(t,q)$-run over $\mathcal{A}[X_1:L_1,\dots,X_n:L_n]$ is defined as a $(t,q)$-run $(T_r,r)$ over $\mathcal{A}$ except that for nodes $x\in T_r$ with $r(x)=(t',X_i)$ for a $t'\in T$, the positive boolean formula $\rho(X_i,l(t'))$ that has to be satisfied by the children of $x$ is replaced by $\true$, if $t'\in L_i$, and by $\false$, if $t'\notin L_i$.
The acceptance of such a path is then defined as before.
By $\mathcal{L}_{q}^\mathcal{T}(\mathcal{A}[X_1:L_1,...,X_n:L_n])$ we denote the set of nodes $t\in T$ such that there is an accepting $(t,q)$-run over $\mathcal{A}[X_1:L_1,\dots,X_n:L_n]$.

For the inductive proof, we now assume that the free variables of the current formula $\psi\in\Sub(\varphi)$ are holes in the automaton and show that the language of this automaton corresponds to the semantics of $\psi$. 
Intuitively, we fill the holes in the automaton, i.e. the free variables of $\psi$, with the same sets of nodes as specified by a given valuation that we consider for the semantics of $\psi$. 
More formally, the holes are filled by sets of tree nodes that correspond to given sets of graph nodes in the valuation.
Therefore, we lift the function $\delta_G \colon V\to T$ for the execution graph $G=(V, l, (\rightarrow^d)_{d\in\moves},\curvearrowright)$ with $\mathcal{T}(G)=\mathcal{T}=(T,r)$ to a function $\tilde{\delta}_G \colon 2^V\to 2^T$ by $\tilde{\delta}_G(A):=\{\delta_G(a)\mid a\in A\}$ and show the following claim:

\underline{Claim:}
For all fixpoint variable assignments $\mathcal{V}$, subformulae $\psi\in \Sub(\varphi)$ with free variables $X_1,\dots,X_n$ and $L_1,\dots,L_n\subseteq V$ we have
\begin{align*}
	\mathcal{L}_\psi^{\mathcal{T}}(\mathcal{\tilde{A}}_\varphi[X_1:\tilde{\delta}_G(L_1),\dots,X_n:\tilde{\delta}_G(L_n)]) = \tilde{\delta}_G(\llbracket \psi\rrbracket_{\mathcal{V}[X_1 \mapsto L_1,\dots, X_n\mapsto L_n]}^G).
\end{align*}
Since $\varphi$ is closed, this implies in particular that
\begin{align*}
	\mathcal{T}\in \mathcal{L}(\mathcal{\tilde{A}}_\varphi)
	\Leftrightarrow \varepsilon\in \mathcal{L}_\varphi^{\mathcal{T}}(\mathcal{\tilde{A}}_\varphi)
	\Leftrightarrow \varepsilon \in \tilde{\delta}_G(\llbracket \varphi\rrbracket^G)
	\Leftrightarrow v_0\in \llbracket \varphi\rrbracket^G
	\Leftrightarrow G \in \llbracket \varphi\rrbracket.
\end{align*}
 We now proceed with the structural induction.
\begin{itemize}
	\item For $\psi\equiv \ap\in \AP$ we have
	\begin{align*}
		&\ \mathcal{L}_\psi^{\mathcal{T}}(\mathcal{\tilde{A}}_\varphi)=\{\delta_G(x)\in T\mid \ap\in l(x)\}\\
		= &\ \tilde{\delta}_G(\{x\in V\mid \ap\in l(x)\})\\
		= &\ \tilde{\delta}_G(\llbracket \psi\rrbracket_{\mathcal{V}}^G).
	\end{align*}
	
	\item For $\psi\equiv \lnot \ap$ with $\ap\in \AP$, the claim follows analogously.
	
	\item For $\psi\equiv X$ and $L\subseteq V$ we have
	\begin{align*}
		&\ \mathcal{L}_\psi^{\mathcal{T}}(\mathcal{\tilde{A}}_\varphi[X:\tilde{\delta}_G(L)])\\
		= &\ \tilde{\delta}_G(L)\\
		= &\ \tilde{\delta}_G(\mathcal{V}[X \mapsto L](X))\\
		= &\ \tilde{\delta}_G(\llbracket \psi\rrbracket_{\mathcal{V}[X\mapsto L]}^G).
	\end{align*}
	
	\item For $\psi\equiv \psi_1\lor \psi_2$, let $X_{k^i_1},\dots,X_{k^i_{n_i}}$ be the free variables of $\psi_i$ for $i\in\{1,2\}$.
	We clearly have \medskip\\
	$\mathcal{L}_{\psi_i}^{\mathcal{T}}(\mathcal{\tilde{A}}_\varphi[X_1:\tilde{\delta}_G(L_1),\dots,X_n:\tilde{\delta}_G(L_n)])
	=\mathcal{L}_{\psi_i}^{\mathcal{T}}(\mathcal{\tilde{A}}_\varphi[X_{k^i_1}:\tilde{\delta}_G(L_{k^i_1}),\dots,X_{k^i_{n_i}}:\tilde{\delta}_G(L_{k^i_{n_i}})])$ \medskip\\
	since a $(t,\psi_i)$-run over $\mathcal{\tilde{A}}_\varphi[X_{k^i_1}:\tilde{\delta}_G(L_{k^i_1}),\dots,X_{k^i_{n_i}}:\tilde{\delta}_G(L_{k^i_{n_i}})]$ can only reach nodes labelled by a fixpoint variable $X$ if $X$ occurs in $\psi_i$.
	Moreover, $\llbracket \psi_i\rrbracket_{\mathcal{V}[X_{k^i_1} \mapsto L_{k^i_1},\dots, X_{k^i_{n_i}}\mapsto L_{k^i_{n_i}}]}^G=\llbracket \psi\rrbracket_{\mathcal{V}[X_1 \mapsto L_1,\dots, X_n\mapsto L_n]}^G$, since all the free variables  $X_{k^i_1},\dots,X_{k^i_{n_i}}$ of $\psi_i$ are also among the free variables $X_1,\dots,X_n$ of $\psi$.
	Thus, we have
	\begin{align*}
		&\ \mathcal{L}_\psi^{\mathcal{T}}(\mathcal{\tilde{A}}_\varphi[X_1:\tilde{\delta}_G(L_1),\dots,X_n:\tilde{\delta}_G(L_n)])\\
		\mybox[IH]{=} & \bigcup\limits_{i\in\{0,1\}}\mathcal{L}_{\psi_i}^{\mathcal{T}}(\mathcal{\tilde{A}}_\varphi[X_{1}:\tilde{\delta}_G(L_{1}),\dots,X_{n}:\tilde{\delta}_G(L_{n})])\\
		\mybox[IH]{=} & \bigcup\limits_{i\in\{0,1\}}\mathcal{L}_{\psi_i}^{\mathcal{T}}(\mathcal{\tilde{A}}_\varphi[X_{k^i_1}:\tilde{\delta}_G(L_{k^i_1}),\dots,X_{k^i_{n_i}}:\tilde{\delta}_G(L_{k^i_{n_i}})])\\
		\oversetbox[IH]{(IH)}{=} & \bigcup\limits_{i\in\{0,1\}}\tilde{\delta}_G(\llbracket \psi_i\rrbracket_{\mathcal{V}[X_{k^i_1} \mapsto L_{k^i_1},\dots, X_{k^i_{n_i}}\mapsto L_{k^i_{n_i}}]}^G)\\
		\mybox[IH]{=} & \bigcup\limits_{i\in\{0,1\}}\tilde{\delta}_G(\llbracket \psi_i\rrbracket_{\mathcal{V}[X_1 \mapsto L_1,\dots, X_n\mapsto L_n]}^G)\\
		\mybox[IH]{=} &\ \tilde{\delta}_G(\bigcup\limits_{i\in\{0,1\}}\llbracket \psi_i\rrbracket_{\mathcal{V}[X_1 \mapsto L_1,\dots, X_n\mapsto L_n]}^G)\\
		\mybox[IH]{=} &\ \tilde{\delta}_G(\llbracket \psi\rrbracket_{\mathcal{V}[X_1 \mapsto L_1,\dots, X_n\mapsto L_n]}^G)
	\end{align*}
	where equation (IH) uses the induction hypothesis.
	\item For $\psi\equiv \psi_1 \land \psi_2$, the claim follows analogously since $\tilde{\delta}_G$ is injective.
	
	\item For $\psi\equiv \bigcirc^f\psi'$ with $f\in\{\globalsucc,\globalpred,\abstractsucc,\caller,\parent,\child\}$, the free variables $X_1,\dots,X_n$ of $\psi$ are also the free variables of $\psi'$.
	Following the definition of the tree successor functions, we easily see that
	\begin{align*}
	& \mathcal{L}_\psi^{\mathcal{T}}(\mathcal{\tilde{A}}_\varphi[X_1:\tilde{\delta}_G(L_1),\dots,X_n:\tilde{\delta}_G(L_n)])\\
	\mybox[*]{=} & \{t\in T\mid\Succ_f^{\mathcal{T}}(t) \text{ is defined and}\\
	&\Succ_f^{\mathcal{T}}(t)\in\mathcal{L}_{\psi'}^{\mathcal{T}}(\mathcal{\tilde{A}}_\varphi[X_1:\tilde{\delta}_G(L_1),\dots,X_n:\tilde{\delta}_G(L_n)])\}\\
	\oversetbox[*]{(IH)}{=} & \{t\in T\mid\Succ_f^{\mathcal{T}}(t) \text{ is defined and } \Succ_f^{\mathcal{T}}(t)\in 	\tilde{\delta}_G(\llbracket \psi'\rrbracket_{\mathcal{V}[X_1 \mapsto L_1,\dots, X_n\mapsto L_n]}^G)\}\\
	\oversetbox[*]{(\textasteriskcentered)}{=} & \{\delta_G(x)\in T\mid \Succ_f^G(x) \text{ is defined and}\\
	&\delta_G(\Succ_f^G(x))\in \tilde{\delta}_G(\llbracket \psi'\rrbracket_{\mathcal{V}[X_1 \mapsto L_1,\dots, X_n\mapsto L_n]}^G)\}\\
	\mybox[*]{=} & \tilde{\delta}_G(\{x\in V\mid \Succ_f^G(x) \text{ is defined and } \Succ_f^G(x)\in \llbracket \psi'\rrbracket_{\mathcal{V}[X_1 \mapsto L_1,\dots, X_n\mapsto L_n]}^G\})\\
	\mybox[*]{=} & \tilde{\delta}_G(\llbracket \psi\rrbracket_{\mathcal{V}[X_1 \mapsto L_1,\dots, X_n\mapsto L_n]}^G),
	\end{align*}
	where equation (IH) uses the induction hypothesis and equation (\textasteriskcentered) uses \cref{lemma:successorequivalence}.
	
	\item
	For $\psi\equiv \bigcirc^{\bar{f}}\psi'$ with $f\in\{\globalsucc,\globalpred,\abstractsucc,\caller,\parent,\child\}$ we observe that a $(t,\psi)$-run behaves as a $(t,\bigcirc^f\psi')$-run except that we move to $\true$ instead of $\false$, if a desired successor or predecessor does not exist.
	Thus, analogously to the previous case, it is easy to see that 
	\begin{align*}
	&\ \mathcal{L}_\psi^{\mathcal{T}}(\mathcal{\tilde{A}}_\varphi[X_1:\tilde{\delta}_G(L_1),\dots,X_n:\tilde{\delta}_G(L_n)])\\
	= &\ \{t\in T\mid\Succ_f^{\mathcal{T}}(t) \text{ is undefined or}\\
	& \Succ_f^{\mathcal{T}}(t)\in \mathcal{L}_{\psi'}^{\mathcal{T}}(\mathcal{\tilde{A}}_\varphi[X_1:\tilde{\delta}_G(L_1),\dots,X_n:\tilde{\delta}_G(L_n)])\}\\
	{=} &\ \tilde{\delta}_G(V\setminus\{x\in V\mid \Succ_f^G(x) \text{ is defined and}\\
	& \Succ_f^G(x)\notin \llbracket \psi'\rrbracket_{\mathcal{V}[X_1 \mapsto L_1,\dots, X_n\mapsto L_n]}^G\})\\
	{=} &\ \tilde{\delta}_G(V\setminus\{x\in V\mid \Succ_f^G(x) \text{ is defined and}\\
	& \Succ_f^G(x)\in \llbracket \lnot\psi'\rrbracket_{\mathcal{V}[X_1 \mapsto L_1,\dots, X_n\mapsto L_n]}^G\})\\
	{=} &\ \tilde{\delta}_G(V\setminus\llbracket \bigcirc^f\lnot\psi'\rrbracket_{\mathcal{V}[X_1 \mapsto L_1,\dots, X_n\mapsto L_n]}^G)\\
	{=} &\ \tilde{\delta}_G(\llbracket \lnot\bigcirc^f\lnot\psi'\rrbracket_{\mathcal{V}[X_1 \mapsto L_1,\dots, X_n\mapsto L_n]}^G)\\
	{=} &\ \tilde{\delta}_G(\llbracket \psi\rrbracket_{\mathcal{V}[X_1 \mapsto L_1,\dots, X_n\mapsto L_n]}^G).
	\end{align*}
	
	\item For $\psi\equiv \mu X.\psi'$, the free variables of $\psi'$ are given by $X_1,\dots, X_n,X$.
	We observe that a $(t,\psi)$-run over $\mathcal{\tilde{A}}_\varphi[X_1:\tilde{\delta}_G(L_1),\dots,X_n:\tilde{\delta}_G(L_n)]$ can only visit states $\varphi'$ of the form $\mu X.\psi''$ or $\nu X.\psi''$ if $\varphi'$ is a subformula of $\psi$.
	Therefore, $\Omega(\psi)$ is the lowest priority occurring in the run so that the state $\psi$ can only be visited finitely often if the run is accepting.
	This means we can characterize $\mathcal{L}_\psi^{\mathcal{T}}(\mathcal{\tilde{A}}_\varphi[X_1:\tilde{\delta}_G(L_1),\dots,X_n:\tilde{\delta}_G(L_n)])$ as the least fixpoint of the function $f:2^T\to 2^T$ with $f(\tilde{\delta}_G(L)):=\mathcal{L}_{\psi'}^{\mathcal{T}}(\mathcal{\tilde{A}}_\varphi[X_1:\tilde{\delta}_G(L_1),\dots,X_n:\tilde{\delta}_G(L_n), X:\tilde{\delta}_G(L)])$.
	By the induction hypothesis, we obtain\\
	$f(\tilde{\delta}_G(L))= \tilde{\delta}_G(\llbracket \psi'\rrbracket_{\mathcal{V}[X_1 \mapsto L_1,\dots, X_n\mapsto L_n,X\mapsto L]}^G)=\tilde{\delta}_G(\alpha_S(L))$ where $S=(G,\\
	\mathcal{V}[X_1 \mapsto L_1,\dots, X_n\mapsto L_n],X,\psi')$.
	By \cref{cor:fixpoint-characterization}, the least fixpoint of $\alpha_S$ is given by $\llbracket \psi\rrbracket^G_{\mathcal{V}[X_1 \mapsto L_1,\dots, X_n\mapsto L_n]}$.
	Since the function $\tilde{\delta}_G$ is trivially monotone and bijective because $\delta_G$ is bijective, we conclude by \cref{lemma:fixpointtranslation} that $\tilde{\delta}_G(\llbracket \psi\rrbracket^G_{\mathcal{V}[X_1 \mapsto L_1,\dots, X_n\mapsto L_n]})$ is the least fixpoint of $f$, i.e.
	\begin{align*}
		\mathcal{L}_\psi^{\mathcal{T}}(\mathcal{\tilde{A}}_\varphi[X_1:\tilde{\delta}_G(L_1),\dots,X_n:\tilde{\delta}_G(L_n)])=\tilde{\delta}_G(\llbracket \psi\rrbracket^G_{\mathcal{V}[X_1 \mapsto L_1,\dots, X_n\mapsto L_n]}).
	\end{align*}
	
	\item For $\psi\equiv \nu X.\psi'$, the claim follows analogously.
	\hfill \qed
\end{itemize}
\end{proof}

In order to prove \cref{thm:dpnautomaton}, we establish the following lemma:

\begin{lemma}\label{lemma:stacklevel}
	Let
	$\mathcal{M}=(S,s_0,\gamma_0,\Delta,L)$
	be a DPN and $G=(V, l, (\rightarrow^d)_{d\in\moves},\curvearrowright)$ be an execution graph of $\mathcal{M}$ witnessed by the assignment $\assignment\colon V\to S\times\Gamma^*\bot$.
	Further, let $x,y\in V$ be nodes with $y=\Succ_\abstractsucc^G(x)$.
	Then there are control locations $s,s'\in S$, stack symbols $\gamma,\gamma'\in \Gamma$ and a stack content $w\in \Gamma^*\bot$ such that $\assignment(x)=(s,\gamma w)$ and $\assignment(y)=(s',\gamma'w)$.
\end{lemma}
\begin{proof}
	Since $y=\Succ_\abstractsucc^G(x)$, we either have $x\rightarrow^\intern y$ or $x\curvearrowright y$.
	In both cases, there is a path $\pi$ from $x$ to $y$ following only $\moves\setminus\{\spawn\}$-transitions such that the number $n$ of $\call$-moves on $\pi$ is equal to the number of $\ret$-moves on $\pi$.
	We show the claim by induction over $n$.
	
	If $n=0$, we have $x\rightarrow^\intern y$. Since $G$ is generated by $\mathcal{M}$, we thus have $\assignment(x)\rightarrow_\intern \assignment(y)$, i.e. there are control locations $s,s'\in S$, stack symbols $\gamma,\gamma'\in \Gamma$ and a stack content $w\in \Gamma^*\bot$ such that $\assignment(x)=(s,\gamma w)$ and $\assignment(y)=(s',\gamma'w)$.
	
	If $n>0$, we have $x\curvearrowright y$.
	Thus, $x$ has a $\call$-successor $x'$ and by \cref{lemma:successorcharacterization}(i), $y$ has a $\ret$-predecessor $y'$ with $\Succ_\caller^G(z')=x$, i.e. there is a path $\pi'$ from $x'$ to $z'$ following abstract successors.
	In particular, in each step in $\pi'$ from a node $u$ to its abstract successor $v$, there is a path $\pi''$ from $u$ to $v$ following only $\moves\setminus\{\spawn\}$-transitions such that the number $m_{\pi''}$ of $\call$-moves on $\pi''$ is equal to the number of $\ret$-moves on $\pi''$.
	Since all of these paths $\pi''$ are proper subpaths of $\pi$, we clearly have $m_{\pi''}<n$ in each step.
	Thus, by the induction hypothesis, there are control locations $s_1,s_2\in S$, stack symbols $\gamma_1,\gamma_2\in \Gamma$ and a stack content $w'\in \Gamma^*\bot$ such that $\assignment(x')=(s_1,\gamma_1 w')$ and $\assignment(y')=(s_2,\gamma_2w')$.
		Since further $x\rightarrow^\call x'$ and hence $\assignment(x)\rightarrow_\call \assignment(x')$, we also have $\assignment(x)=(s,\gamma w)$ for some $s\in S$ and $\gamma\in\Gamma$ as well as $w'=\gamma'w$ for some $\gamma'\in\Gamma$.
		Finally, since $y'\rightarrow^\ret y$ and hence $\assignment(y')\rightarrow_\ret \assignment(y)$, we also have $\assignment(y)=(s',w')=(s',\gamma'w)$ for some $s'\in S$.
		\hfill \qed
\end{proof}

\begin{proof}[\textbf{Proof of \cref{thm:dpnautomaton}}]
	We first show that $\mathcal{{A}}_\mathcal{M}$ accepts all tree representations of executions graphs generated by $\mathcal{M}$. 
	
	For this, let $G=(V, l, (\rightarrow^d)_{d\in\moves},\curvearrowright)$ be an execution graph of $\mathcal{M}$ witnessed by the assignment $\assignment\colon V\to S\times\Gamma^*\bot$.
	Further, let $\mathcal{T}(G)=(T,r)$ be the tree representation of $G$.
	We first define maps $r_S\colon V\to S$ and $r_\Gamma\colon V\to \Gamma$ as follows.
	For nodes $x\in V$ with $\assignment(x)=(s,\gamma w)$ for some stack symbol $\gamma\in\Gamma$ and stack content $w\in \Gamma^*\bot$ we set $r_S(x):=s$ and $r_\Gamma(x)=\gamma$.
	Moreover, we inductively define a map $r_R\colon V\to (S\times \Gamma)\cup \{\bot\}$ as follows.
	First, we set $r_R(v_0):=\bot$.
	Then, for all nodes $x\in V$,
	\begin{itemize}
		\item if there is a node $y\in V$ such that $y$ is an $\intern$-predecessor of $x$ or $y\curvearrowright x$ (the latter holds by \cref{lemma:successorcharacterization}(i) iff $x$ has a $\ret$-predecessor), we set $r_R(x):=r_R(y)$,
		\item if $x$ (i) has a $\spawn$-predecessor or (ii) it has a $\call$-predecessor $y$ and there is no $z\in V$ with $y\curvearrowright z$, we set $r_R(x):=\bot$ and
		\item if $x$ has a $\call$-predecessor $y$ and there is $z\in V$ with $y\curvearrowright z$, we set $r_R(x):=(r_S(z'), r_\Gamma(z'))$, where $z'$ is the $\ret$-predecessor of $z$ (this node exists in this case by \cref{lemma:successorcharacterization}(i)).
	\end{itemize} 
	Finally, we define a map $r_A \colon T\to Q$ by $r_A(\delta_G(x)):=(r_S(x),r_\Gamma(x),r_R(x))$ and show the following claim:
	
	\underline{Claim:} $r_A$ is an accepting $(\varepsilon, q_0)$-run of $\mathcal{A}_\mathcal{M}$ over $\mathcal{T}(G)$.
	
	First, we show that $r_A$ is an $(\varepsilon, q_0)$-run of $\mathcal{A}_\mathcal{M}$ over $\mathcal{T}(G)$.
	
	For the initial node, we have $\assignment(v_0)=(s_0,\gamma_0\bot)$ as well as $r_R(v_0)=\bot$ and thus $r_A(\varepsilon)=r_A(\delta_G(v_0))=(s_0,\gamma_0,\bot)=q_0$.
		
	Now let $t=\delta_G(x)\in T$ be a node with $r(t)=(l(x),d(x),p(x))$ and $r_A(t)=(s,\gamma,c)$.
	Then we have $\assignment(x)=(s,\gamma w)$ for a stack content $w\in \Gamma^*\bot$ and the node is labelled by $l(x)=L(s,\gamma)$ since $G$ is generated by $\mathcal{M}$.
	By a case distinction on $d(x)$, we show that the successors of $t$ satisfy the transition function of $\mathcal{A}_{\mathcal{M}}$.
	\begin{itemize}
		\item 
		If $d(x)=\intern$, then $x$ has an $\intern$-successor $y$ with $\assignment(x)\rightarrow_\intern \assignment(y)$.
		Thus, there are $s'\in S$ and $\gamma'\in \Gamma$ with $\assignment(y)=(s',\gamma'w)$ and $s\gamma\rightarrow s'\gamma'\in \Delta_I$.
		Moreover, we have $r_R(y)=r_R(x)=c$ and $\delta_G(y)=t\cdot 0$.
		Thus, $\{(0,r_A(t\cdot 0))\}=\{(0, (s',\gamma',c))\}$ satisfies $\rho((s,\gamma,c),r(t))$.
		 
		\item 
		If $d(x)=\call$, then $x$ has a $\call$-successor $y$ with $\assignment(x)\rightarrow_\call \assignment(y)$.
		Thus, there are $s'\in S$ and $\gamma',\gamma''\in \Gamma$ with $\assignment(y)=(s',\gamma'\gamma''w)$ and $s\gamma\rightarrow s'\gamma'\gamma''\in \Delta_C$.
		Moreover, there is no node $z\in V$ with $x\curvearrowright z$, i.e. $r_R(x)=\bot$, and $\delta_G(y)=t\cdot 0$.
		Thus $\{(0,r_A(t\cdot 0))\}=\{(0, (s',\gamma',\bot))\}$ satisfies $\rho((s,\gamma,c),r(t))$.
		
		\item 
		If $d(x)=\callret$, then $x$ has a $\call$-successor $y$ with $\assignment(x)\rightarrow_\call \assignment(y)$.
		Thus, there are $s'\in S$ and $\gamma',\gamma''\in \Gamma$ with $\assignment(y)=(s',\gamma'\gamma''w)$ and $s\gamma\rightarrow s'\gamma'\gamma''\in \Delta_C$.
		Moreover, there is a node $z\in V$ with $x\curvearrowright z$.
		By \cref{lemma:successorcharacterization}(i), $z$ has a $\ret$-predecessor $z'$ with $\Succ_\caller^G(z')=x$, i.e. there is a path from the $\call$-successor $y$ of $x$ to $z'$ following abstract successors.
		Hence, by \cref{lemma:stacklevel}, there are $s_r\in S$ and $\gamma_r\in\Gamma$ with $\assignment(z')=(s_r,\gamma_r\gamma''w)$.
		Moreover, since $z'\rightarrow^\ret z$, we have $\assignment(z')\rightarrow_\ret \assignment(z)$, i.e. there is $s''\in S$ with $\assignment(z)=(s'',\gamma''w)$ and $s_r\gamma_r\rightarrow s''\in\Delta_R$.
		Thus, we have $r_R(y)=(r_S(z'),r_\Gamma(z'))=(s_r,\gamma_r)$, $\delta_G(y)=t\cdot 0$, $r_R(z)=r_R(x)=c$ and $\delta_G(z)=t\cdot 1$.
		Hence, $\{(0,r_A(t\cdot 0)), (1, r_A(t\cdot 1))\}=\{(0, (s',\gamma',(s_r,\gamma_r))), (1, (s'',\gamma'',c))\}$ satisfies $\rho((s,\gamma,c),r(t))$.
		
		\item 
		If $d(x)=\spawn$, then $x$ has an $\intern$-successor $y$ and a $\spawn$-successor $z$ with $\assignment(x)\rightarrow \assignment(y)\vartriangleright \assignment(z)$.
		Thus, there are $s',s_n\in S$ and $\gamma',\gamma_n\in \Gamma$ with $\assignment(y)=(s',\gamma'w)$, $\assignment(z)=(s_n,\gamma_n\bot)$ and $s\gamma\rightarrow s'\gamma'\vartriangleright s_n\gamma_n\in \Delta_S$.
		Moreover, we have $r_R(y)=r_R(x)=c$, $\delta_G(y)=t\cdot 0$, $r_R(z)=\bot$, and $\delta_G(z)=t\cdot 1$.
		Thus, $\{(0,r_A(t\cdot 0)), (1, r_A(t\cdot 1))\}=\{(0, (s',\gamma',c)), (1, (s_n, \gamma_n, \bot))\}$ satisfies $\rho((s,\gamma,c),r(t))$.
		
		\item 
		If $d(x)=\ret$, then $x$ has a $\ret$-successor $y$ with $\assignment(x)\rightarrow_\ret \assignment(y)$. Thus, there is $s'\in S$ with $\assignment(y)=(s',w)$ and $s\gamma\rightarrow s'\in \Delta_R$.
		Moreover, by \cref{lemma:successorcharacterization}(i), there is a node $z\in V$ with $z\curvearrowright y$ and $z=\Succ_\caller^G(x)$, i.e. there is a path from the $\call$-successor $z'$ of $z$ to $x$ following abstract successors.
		Thus, we have $c=r_R(x)=r_R(z')=(r_S(x),r_\Gamma(x))=(s,\gamma)$, i.e. $\emptyset$ satisfies $\true=\rho((s,\gamma,(s,\gamma)),r(t))=\rho((s,\gamma,c),r(t))$.
		
		\item 
		If $d(x)=\labelend$, then $x$ has no successors and hence $\assignment(x)$ has no successor,
		i.e. there is no rule for $s\gamma$ in $\Delta$.
		Assume towards contradiction that $c\neq \bot$.
		By construction, there must be nodes $y,z\in V$ with $y\curvearrowright z$ and a path from the $\call$-successor $y'$ of $y$ to $x$ following abstract successors.
		Thus, $\delta_G(x)$ is the $\{\intern,\ret\}$-descendant leaf of the left child $\delta_G(y')=\delta_G(y)\cdot 0$ of the parent node $\delta_G(y)$ of $\delta_G(z)=\delta_G(y)\cdot 1$.
		Then we have $\Succ_\globalpred^{\mathcal{T}(G)}(\delta_G(z))=\delta_G(x)$.
		Using \cref{lemma:successorequivalence}, we obtain 
		$\delta_G(\Succ_\globalpred^G(z))=\Succ_\globalpred^{\mathcal{T}(G)}(\delta_G(z))=\delta_G(x)$ and thus $\Succ_\globalpred^G(z)=x$ since $\delta_G$ is injective.
		This means that $z$ is a successor of $x$, which contradicts our assumption that $x$ has no successor.
		Thus, we have $c=\bot$ and hence $\emptyset$ satisfies $\true=\rho((s,\gamma,\bot),r(t))=\rho((s,\gamma,c),r(t))$.
	\end{itemize}
	Thus, we have established that $r_A$ is an $(\varepsilon, q_0)$-run of $\mathcal{{A}}_\mathcal{M}$ over $\mathcal{T}(G)$.
	Since we have $\Omega(q)=0$ for all $q\in Q$, the run is clearly accepting, and thus $\mathcal{T}(G)\in\mathcal{L}(\mathcal{{A}}_\mathcal{M})$.
	
	For the other direction of the theorem, we show that all exeution trees accepted by $\mathcal{A}_\mathcal{M}$ are tree representations of execution graphs of $\mathcal{M}$.
	For this, let $G=(V, l, (\rightarrow^d)_{d\in\moves},\curvearrowright)$ be an execution graph with $\mathcal{T}(G)=(T,r)\in\mathcal{L}(\mathcal{A}_\mathcal{M})$.
	Then there is an accepting $(\varepsilon, q_0)$-run $r_A$ of $\mathcal{{A}}_\mathcal{M}$ over $\mathcal{T}(G)$. 
	We now inductively define an assignment $\assignment\colon V\to S\times\Gamma^*\bot$ such that for all nodes $v\in V$ with $\assignment(v)=(s,\gamma w)$ for a stack symbol $\gamma\in\Gamma$ and a stack content $w\in\Gamma^*\bot$, we have $r_A(\delta_G(v))=(s,\gamma,c)$ for a $c\in (S\times \Gamma)\cup \{\bot\}$. 
	
	In the base case, we set $\assignment(v_0):=(s_0,\gamma_0\bot)$.
	Since $r_A$ is an $(\varepsilon, q_0)$-run of $\mathcal{{A}}_\mathcal{M}$, we have $r_A(\delta_G(v_0))=r_A(\varepsilon)=q_0=(s_0,\gamma_0,\bot)$.
	
	In the inductive step, let $\assignment$ be defined for some node $v\in V$ with $r(\delta_G(v))=(l',d,p)$ such that there are $\gamma\in\Gamma$ and $w\in\Gamma^*\bot$ with $\assignment(v)=(s,\gamma w)$ and $r_A(\delta_G(v))=(s,\gamma,c)$ for a $c\in (S\times \Gamma)\cup \{\bot\}$.
	We define the mapping for the successors of $v$ by a case distinction on $d$. 
	\begin{itemize}
		\item
		If $d=\intern$, $v$ has exactly one $\intern$-successor $v'$ and we have $r_A(\delta_G(v'))=r_A(\delta_G(v)\cdot 0)=(s',\gamma',c)$ for some $s'\in S$ and $\gamma'\in\Gamma$ such that $s\gamma\rightarrow s'\gamma'\in\Delta_I$.
		Then we set $\assignment(v'):=(s',\gamma' w)$.
		
		\item
		If $d=\call$, $v$ has exactly one $\call$-successor $v'$ and we have $r_A(\delta_G(v'))=r_A(\delta_G(v)\cdot 0)=(s',\gamma',\bot)$ for some $s'\in S$ and $\gamma'\in\Gamma$ such that $s\gamma\rightarrow s'\gamma'\gamma''\in\Delta_C$ for a $\gamma''\in\Gamma$.
		Then we set $\assignment(v'):=(s',\gamma'\gamma''w)$.
		
		\item
		If $d=\callret$, $v$ has exactly one $\call$-successor $v'$ and there is $v''\in V$ with $v\curvearrowright v''$.
		Then we have $r_A(\delta_G(v'))=r_A(\delta_G(v)\cdot 0)=(s',\gamma',(s_r,\gamma_r))$ and 
		$r_A(\delta_G(v''))=r_A(\delta_G(v)\cdot 1)=(s'',\gamma'',c)$
		for some $s',s'',s_r\in S$ and $\gamma',\gamma'',\gamma_r\in\Gamma$ such that $s\gamma\rightarrow s'\gamma'\gamma''\in\Delta_C$ and $s_r\gamma_r\rightarrow s''\in\Delta_R$.
		Then we set $\assignment(v'):=(s',\gamma'\gamma''w)$ and $\assignment(v''):=(s'',\gamma''w)$.
		
		\item
		If $d=\spawn$, $v$ has exactly one $\intern$-successor $v'$ and one $\spawn$-successor $v''$.
		Then we have $r_A(\delta_G(v'))=r_A(\delta_G(v)\cdot 0)=(s',\gamma',c)$ and $r_A(\delta_G(v''))=r_A(\delta_G(v)\cdot 1)=(s_n,\gamma_n,\bot)$	for some $s',s_n\in S$ and $\gamma',\gamma_n\in\Gamma$ such that $s\gamma\rightarrow s'\gamma'\vartriangleright s_n\gamma_n\in\Delta_S$.
		Then we set $\assignment(v'):=(s',\gamma'w)$ and $\assignment(v''):=(s_n,\gamma_n\bot)$.
	\end{itemize}
	Most conditions for $\assignment$ to witness that $G$ is generated by $\mathcal{M}$ are directly satisfied by construction.
	We only show the more involved conditions regarding nodes with $\ret$-successors and nodes without successors.
	In order to do this, we inductively show that the following holds for all nodes $v\in V$ with $r_A(\delta_G(v))=(s,\gamma,c)$:
	
	$(*)$ We have $c=(s_r,\gamma_r)$ for a $s_r\in S$ and $\gamma_r\in\Gamma$ iff there are nodes $v_1,v_2,v_3\in V$ with $v_1\rightarrow^\call v_2$, $v_1\curvearrowright v_3$ and $s'\in S$ and $\gamma'\in \Gamma$ with $r_A(\delta_G(v_2))=(s',\gamma',c)$ and there is a path from $v_2$ to $v$ following abstract successors.
	
	In the base case, where $v=v_0$, we have $c=\bot$ and there is no incoming $\intern$-transition or nesting edge to $v$.
	
	In the inductive step, let $v\in V$ be a node with $r_A(\delta_G(v))=(s,\gamma,c)$ such that $\delta_G(v)=\delta_G(v')\cdot i$ for an $i\in\{0,1\}$ and a node $v'\in V$  with $r(\delta_G(v'))=(l',d,p)$ and $r_A(\delta_G(v'))=(s',\gamma',c')$.
	We show $(*)$ for $v$ by a case distinction on $d$.
	\begin{itemize}
		\item If $d=\intern$, we have $i=0$ and $\rho(r_A(\delta_G(v')),r(\delta_G(v')))$ is a disjunction of formulae of the form $(0,(s'',\gamma'',c'))$ for $s''\in S$ and $\gamma''\in \Gamma$, i.e. $c=c'$.
		Since $v'\rightarrow^\intern v$ and abstract successors are uniquely determined, the right side of the equivalence in $(*)$ is satisfied for $v$ iff it is satisfied for $v'$.
		Thus, by the induction hypothesis, $(*)$ also holds for $v$.
		
		\item If $d=\call$, we have $i=0$ and $\rho(r_A(\delta_G(v')),r(\delta_G(v')))$ is a disjunction of formulae of the form $=(0,(s'',\gamma'',\bot))$  for $s''\in S$ and $\gamma''\in \Gamma$, i.e. $c=\bot$.
		Since there is no node $\tilde{v}\in V$ with $\tilde{v}\rightarrow^\intern v$, $\tilde{v}\curvearrowright v$ or $v'\curvearrowright \tilde{v}$, the required nodes and the path for any $c\neq \bot$ do not exist for $v$.
		
		\item If $d=\callret$, then $\rho(r_A(\delta_G(v')),r(\delta_G(v')))$ is a disjunction of formulae of the form  $(0,(s_1,\gamma_1,(s_r,\gamma_r)))\land (1,(s_2,\gamma_2,c'))$ for $s_1,s_2,s_r\in S$ and $\gamma_1,\gamma_2,\gamma_r\in \Gamma$.
		If additionally $i=0$, then $c=(s_r,\gamma_r)$ for some $s_r\in S$ and $\gamma_r\in \Gamma$.
		Since  $v'\rightarrow^\call v$ and $v'\curvearrowright v''$ for some $v''\in V$,
		the nodes $v_1=v'$, $v_2=v$ and $v_3=v''$ and the empty path from $v_2=v$ to $v$ witness that $(*)$ holds.
		If instead $i=1$, then $c=c'$ and $v'\curvearrowright v$.
		Since abstract successors are uniquely determined, the right side of the equivalence in $(*)$ is satisfied for $v$ iff it is satisfied for $v'$.
		Thus, by the induction hypothesis, $(*)$ also holds for $v$.
		
		\item If $d=\spawn$, then $\rho(r_A(\delta_G(v')),r(\delta_G(v')))$ is a disjunction of formulae of the form $(0,(s_1,\gamma_1,c'))\land (1,(s_2,\gamma_2,\bot))$ for $s_1,s_2\in S$ and $\gamma_1,\gamma_2\in \Gamma$.
		If additionally $i=0$, then $c=c'$.
		Since $v'\rightarrow^\intern v$ and abstract successors are uniquely determined, the right side of the equivalence in $(*)$ is satisfied for $v$ iff it is satisfied for $v'$.
		Thus, by induction hypothesis, $(*)$ also holds for $v$.
		If instead $i=1$, then $c=\bot$.
		Since there is no node $\tilde{v}\in V$ with $\tilde{v}\rightarrow^\intern v$, $\tilde{v}\curvearrowright v$ or $v'\curvearrowright \tilde{v}$, the required nodes and the path for any $c\neq \bot$ do not exist for $v$.
	\end{itemize}

	Given $(*)$, we now show the more involved conditions regarding nodes with $\ret$-successors.
	For this, let $x,y\in V$ be nodes with $x\rightarrow^\ret y$.
	By \cref{lemma:successorcharacterization}(i), there is a node $z\in V$ with $z\curvearrowright y$ and $\Succ_\caller^G(x)=z$, i.e. there is a path from the $\call$-successor $z'$ of $z$ to $x$ following abstract successors.
	Let $r_A(\delta_G(x))=(s_x,\gamma_x,c_x)$ and $r_A(\delta_G(z))=(s,\gamma,c)$.
	Then we have $\assignment(z)=(s,\gamma w)$ for a $w\in\Gamma^*\bot$.
	Since $r(\delta_G(z))=(l(z),\callret,p(z))$,  $\rho(r_A(\delta_G(z)),r(\delta_G(z)))$ is a disjunction of formulae of the form  $(0,(s',\gamma',(s_r,\gamma_r)))\land (1,(s'',\gamma'',c))$ with $s\gamma\rightarrow s'\gamma'\gamma''\in \Delta_C$ and $s_r\gamma_r\rightarrow s''\in \Delta_R$.
	Thus, $r_A(\delta_G(z'))=(s',\gamma',(s_r,\gamma_r))$ as well as $r_A(\delta_G(y))=(s'',\gamma'',c)$ for some $s',s'',s_r\in S$ and $\gamma',\gamma'',\gamma_r\in\Gamma$ with $s\gamma\rightarrow s'\gamma'\gamma''\in \Delta_C$ and $s_r\gamma_r\rightarrow s''\in \Delta_R$.
	By $(*)$ we infer $c_x=(s_r,\gamma_r)$.
	Moreover, since $x$ has a $\ret$-successor, we have $d(x)=\ret$ and thus $\rho(r_A(\delta_G(x)),r(\delta_G(x)))$ is $\true$, if $s_x=s_r$ and $\gamma_x=\gamma_r$, and $\false$ otherwise.
	Since the transition function must be satisfied by the children of $\delta_G(x)$, we thus have $s_x=s_r$ and $\gamma_x=\gamma_r$, i.e. $\assignment(x)=(s_r,\gamma_r w')$ for a $w'\in \Gamma^*\bot$.
	By construction of $\assignment$, we clearly have $w'=\gamma''w$.
	Thus, since $\assignment(x)=(s_r,\gamma_r\gamma'' w)$, $\assignment(y)=(s'',\gamma''w)$ and $s_r\gamma_r\rightarrow s''\in \Delta_R$, we have $\assignment(x)\rightarrow_\ret \assignment(y)$.
	
	Finally, let $x\in V$ be a node without successors.
	Then we have $r(\delta_G(x))=(l(x),\labelend,p(x))$.
	Let $r_A(\delta_G(x))=(s,\gamma,c)$, i.e. $\assignment(x)=(s,\gamma w)$ for some $w\in\Gamma^*\bot$.
	Assume towards contradiction that we have $c=(s_r,\gamma_r)$ for some $s_r\in S$ and $\gamma_r\in\Gamma$.
	By $(*)$, there are $v_1,v_2,v_3\in V$ with $v_1\rightarrow^\call v_2$ and $v_1\curvearrowright v_3$ and there is a path from $v_2$ to $x$ following abstract successors. 
	Thus, $\delta_G(x)$ is the $\{\intern,\ret\}$-descendant leaf of the left child $\delta_G(v_2)=\delta_G(v_1)\cdot 0$ of the parent node $\delta_G(v_1)$ of $\delta_G(v_3)=\delta_G(v_1)\cdot 1$.
	Then we have $\Succ_\globalpred^{\mathcal{T}(G)}(\delta_G(v_3))=\delta_G(x)$.
	Using \cref{lemma:successorequivalence}, we obtain 
	$\delta_G(\Succ_\globalpred^G(v_3))=\Succ_\globalpred^{\mathcal{T}(G)}(\delta_G(v_3))=\delta_G(x)$ and thus $\Succ^G_\globalpred(v_3)=x$ since $\delta_G$ is injective. 
	This means that $v_3$ is a successor of $x$, which contradicts our assumption that $x$ has no successor.	
	Thus, we must have $c=\bot$, i.e. $\rho(r_A(\delta_G(x)),r(\delta_G(x)))$ is $\true$, if
	there is no rule for $s\gamma$ in $\Delta$,
	and $\false$ otherwise.
	Since the transition function must be satisfied by the children of $\delta_G(x)$, there is thus no successor of $\assignment(x)$.
	\hfill \qed
\end{proof}
\end{document}